\documentclass[onecolumn,american,english,final, twocolumn]{IEEEtran}
\usepackage[T1]{fontenc}
\usepackage[latin9]{inputenc}
\usepackage{mathrsfs}
\usepackage{bm}
\usepackage{amsmath}
\usepackage{amsthm}
\usepackage{amssymb}
\usepackage{graphicx}

\makeatletter
\theoremstyle{plain}
\newtheorem{lem}{\protect\lemmaname}
\theoremstyle{plain}
\newtheorem{thm}{\protect\theoremname}
\theoremstyle{plain}
\newtheorem{cor}{\protect\corollaryname}
\theoremstyle{plain}
\newtheorem{prop}{\protect\propositionname}

\usepackage{amsmath}
\usepackage{amssymb}
\usepackage{multicol}
\usepackage{algorithm}
\usepackage{algorithmic}
\usepackage{array}
\usepackage{booktabs}

\usepackage[level=3]{wgroup_message}  
\usepackage{graphicx,psfrag,cite,subfigure}

\author{

\IEEEauthorblockN{Haochen~Xu$^{*\ddagger}$, Junting~Chen$^{*\ddagger}$, and Pooi-Yuen~Kam$^{\dagger\ddagger}$}

\IEEEauthorblockA{
$^{*}$School of Science and Engineering,
$^{\dagger}$School of Artificial Intelligence,\\
$^{\ddagger}$Shenzhen Future Network of Intelligence Institute (FNii-Shenzhen),\\
The Chinese University of Hong Kong, Shenzhen, Guangdong 518172, China}
}


\usepackage[acronym]{glossaries}
\newcommand{\newac}{\newacronym}
\newcommand{\ac}{\gls}

\newcommand{\acpl}{\glspl}

\makeglossaries

\newac{dft}{DFT}{Discrete Fourier Transform}
\newac{hmm}{HMM}{hidden Markov model}
\newac{ris}{RIS}{Reconfigurable intelligent surface}
\newac{awgn}{AWGN}{additive white Gaussian noise}
\newac{aod}{AOD}{angle of departure}
\newac{aoa}{AOA}{angle of arrival}
\newac{bs}{BS}{base station}
\newac{ue}{UE}{user}
\newac{mle}{MLE}{maximum likelihood estimation}
\newac{snr}{SNR}{signal-to-noise ratio}
\newac{csi}{CSI}{channel state information}
\newac{sbl}{SBL}{sparse Bayesian learning}
\newac{siso}{SISO}{Single Input Single Output}
\newac{hpbw}{HPBW}{Half Power Beam Width}
\newac{upa}{UPA}{uniform planar array}
\newac{mab}{L2O}{learning to optimize}
\newac{rf}{RF}{radio frequency}
\newac{mimo}{MIMO}{multiple input multiple output}
\newac{miso}{MISO}{multiple input single output}
\newac{los}{LOS}{line-of-sight}
\newac{urp}{URP}{unique representation property}
\newac{em}{EM}{electromagnetic}
\newac{sv}{SV}{Saleh-Valenzuela}
\newac{ts}{TS}{Thompson sampling}
\newac{ucb}{UCB}{upper confidence bound}
\newac{mmwave}{MMWave}{millimeter wave}
\newac{ula}{ULA}{uniform linear array}
\newac{mmse}{MMSE}{minimum mean square error}
\newac{nmse}{NMSE}{normalized mean squared error}
\newac{ls}{LS}{least square}
\newac{ofdm}{OFDM}{orthogonal frequency division multiplexing}
\newac{wlog}{w.l.o.g}{without loss of generality}
\newac{ecsm}{ECSM}{enhanced conditional sampled mean}
\newac{pdf}{PDF}{probability density function}
\setkeys{Gin}{width=1.0\columnwidth}

\usepackage{geometry}
\providecommand{\lyxdeleted}[3]{{\color{lyxdeleted}{}}}
\IEEEoverridecommandlockouts

\makeatother

\usepackage{babel}
\addto\captionsamerican{\renewcommand{\corollaryname}{Corollary}}
\addto\captionsamerican{\renewcommand{\lemmaname}{Lemma}}
\addto\captionsamerican{\renewcommand{\propositionname}{Proposition}}
\addto\captionsamerican{\renewcommand{\theoremname}{Theorem}}
\addto\captionsenglish{\renewcommand{\corollaryname}{Corollary}}
\addto\captionsenglish{\renewcommand{\lemmaname}{Lemma}}
\addto\captionsenglish{\renewcommand{\propositionname}{Proposition}}
\addto\captionsenglish{\renewcommand{\theoremname}{Theorem}}
\providecommand{\corollaryname}{Corollary}
\providecommand{\lemmaname}{Lemma}
\providecommand{\propositionname}{Proposition}
\providecommand{\theoremname}{Theorem}

\begin{document}
\title{Bayesian Bandit Beamforming with Implicit Channel Learning for RIS
under Hybrid Near/Far-Field Propagation
}\maketitle
\begin{abstract}
\acpl{ris} can improve high-frequency wireless links by shaping the
propagation environment, but their passive architecture makes channel
acquisition costly. Conventional estimate-then-optimize methods usually
require pilot overhead that scales with the number of reflecting elements,
which is undesirable under short coherence times and hybrid near-/far-field
propagation. This paper proposes the Bayesian bandit framework for
\ac{ris} phase-shift configuration with implicit channel learning.
The method updates a Gaussian posterior of the cascaded channel from
one scalar pilot observation per slot and uses Thompson sampling to
balance channel estimation and beamforming gain. We derive a Bayesian
regret decomposition that connects Bayesian received-power regret
to posterior uncertainty contraction, and further establish a conditional
sublinear Bayesian-regret guarantee. To exploit sparse hybrid-field
propagation, we develop an energy-focusing angle-distance dictionary
and a \ac{sbl}-based Thompson-sampling algorithm with warm-started
hyperparameter refinement. Simulations show that the proposed policies
approach the perfect-\ac{csi} benchmark in the \ac{los}-dominant
setting within 10 time block and improve transmission efficiency over
the considered baselines in multipath and Rayleigh fading scenarios.
\end{abstract}


\section{Introduction}

High-frequency wireless communications are expected to support data-intensive
services in future 6G networks by exploiting a large bandwidth \cite{Miao2025}.
However, high-frequency links suffer from severe path loss, blockage
sensitivity, and limited coverage. \acpl{ris}, which consist of a
large number of passive reflecting elements, provide a cost effective
means to reshape the wireless propagation environment and enhance
link quality \cite{Du2024}. By properly configuring the \ac{ris}
phase shifts, the signals reflected by different elements can be coherently
combined at the receiver. Despite this potential, RIS phase-shift
configuration remains a challenging signal processing problem in wireless
communication \cite{zhao2021}.

A major challenge in RIS phase-shift configuration is channel acquisition
from passive and high-dimensional measurements. Early studies typically
assume that accurate \ac{csi} is available, and then optimize the
active and passive beamformers using semidefinite relaxation, alternating
optimization, manifold optimization, or related non-convex optimization
methods \cite{SDR,Arora2022,MO}. These works demonstrate the potential
gains of \ac{ris}-aided transmission, but the perfect-\ac{csi} assumption
is difficult to satisfy in practice. Since a passive \ac{ris} has
no radio-frequency chains or baseband processing capability, it cannot
directly observe the incident signal. The cascaded \ac{bs}-\ac{ris}-user
channel must therefore be inferred indirectly from pilot measurements
collected at the user or the \ac{bs}. The resulting pilot overhead
usually scales with the number of \ac{ris} elements, making channel
acquisition a high dimensional measurement problem with prohibitive
overhead for large-scale RISs and short channel coherence times.

A large number of works reduce this overhead by exploiting additional
channel structures. Element-grouping methods let adjacent RIS elements
share the same reflection coefficient, thereby reducing the number
of unknown channel coefficients \cite{group}. However, grouping lowers
the spatial resolution of the RIS and may sacrifice beamforming gain.
Compressive-sensing-based methods exploit angular or polar-domain
sparsity in millimeter-wave channels and recover the cascaded channel
from fewer pilots using sparse or structured recovery algorithms \cite{RIS_CS1,cao2026,RIS_CS2}.
Multiuser methods exploit shared channel structure among users and
perform joint channel inference through matrix calibration or message
passing \cite{xia2021}. Parametric methods further estimate angles
and distances using short pilot sequences in near-field and far-field
RIS channels \cite{MLE2}. These methods improve pilot efficiency,
but they still follow an estimate-then-optimize paradigm: a dedicated
training stage is first used for CSI acquisition, and data transmission
is performed only after channel estimation. Note that given a probed
RIS configuration in the training stage, only a small portion of the
bandwidth is sufficient to estimate the equivalent RIS channel. However,
the remaining bandwidth may not be efficiently used for data payload
transmission because the probed RIS configuration is not optimized
for data transmission, resulting in a waste of time-frequency resources
in the training stage.

Beyond channel estimation-based designs, beam control can also be
realized through dedicated hardware \cite{engineering3} and signal
processing methods. Specifically, beam-training and blind beamforming
methods avoid explicit full-channel estimation by searching for a
good reflection pattern. Exhaustive beam sweeping is simple but incurs
large overhead when the codebook is large \cite{Wang2022}. Hierarchical
beam searches reduce this cost by progressively refining the beam
direction \cite{hibeamtraining}. Blind beamforming methods configure
RIS phase shifts based on received signal strength and can provide
performance guarantees under suitable assumptions \cite{blind2}.
Learning-based methods, including unsupervised learning, reinforcement
learning, and temporal tracking approaches, further reduce online
search overhead by exploiting training data, side information, or
channel dynamics \cite{Yang2021,Ouyang2023,xia2024learning,Han2025}. 

Among online beam control methods, bandit-based beam alignment uses
online feedback to balance exploration and exploitation over predefined
beam codebooks \cite{TSSBL,Bandit1}. Adaptive Thompson sampling has
also been studied for \ac{miso} beam tracking, where an active \ac{bs}
precoder is selected to jointly support \ac{csi} acquisition and
data transmission \cite{xu2024adaptive}. These methods are effective
when the beam space is well represented by a finite codebook, or when
reliable training data or temporal models are available. However,
they do not directly address the passive RIS observation structure:
different from active BS precoding in \ac{miso} systems, each selected
RIS phase vector is not merely a beam index, but also a sensing vector
that produces a scalar observation of the underlying cascaded channel.
Moreover, the same RIS configuration should support both channel probing
and payload transmission in a slot.

Another challenge comes from large-aperture RIS propagation. When
the physical aperture of the RIS becomes large relative to the wavelength,
the Rayleigh distance can extend to tens or even hundreds of meters,
making spherical-wave near-field propagation non-negligible \cite{Guerra2021distance}.
Some scatterers may lie in the near-field of the RIS, where both angle
and distance are needed to describe the spherical wavefront, while
others may remain in the far-field, where angular information is sufficient.
Far-field angular dictionaries are compact but inaccurate for near-field
components, whereas polar-domain near-field dictionaries improve representation
accuracy at the cost of a larger dictionary size \cite{3dB_coherence,columncovariance}.
For online RIS beamforming, the dictionary must not only represent
hybrid near-/far-field propagation accurately, but also remain compact
enough for fast posterior updates and Thompson sampling.

The above discussion reveals two coupled bottlenecks: first, RIS phase
shifts should be learned without reserving many slots solely for channel
training; second, hybrid-field sparsity should be exploited without
creating an overly large dictionary. This paper addresses these challenges
by developing a Bayesian bandit framework for RIS beamforming with
implicit channel learning. Instead of separating channel estimation
and data transmission, in each slot, the selected RIS phase-shift
vector is used to obtain a scalar pilot observation, update a Bayesian
posterior of the cascaded channel, and support payload transmission
under the same RIS configuration. This formulation naturally captures
the exploration-exploitation trade-off: the RIS should probe uncertain
channel directions to improve future decisions, while maintaining
high data transmission efficiency.

Such a framework is nontrivial for three reasons. First, the RIS action
is continuous and constrained by unit-modulus phase shifts, unlike
conventional finite-arm bandit models. Second, the information gained
in each slot depends on the RIS configuration selected by the learning
policy itself, which couples posterior contraction and regret analysis
with the beamforming decisions. Third, under hybrid near-/far-field
propagation, the channel representation must be accurate, sparse,
and computationally manageable. To address these issues, we develop
a Gaussian posterior-based Thompson-sampling method for implicit channel
learning, derive a Bayesian regret decomposition that connects Bayesian
received-power regret to posterior uncertainty contraction, and design
an energy-focusing angle-distance dictionary together with a \ac{sbl}-based
Thompson-sampling algorithm.

\begin{itemize}
\item We formulate RIS phase-shift configuration under a shared pilot-data
architecture as a Bayesian bandit beamforming problem. Under an effective
Gaussian prior of the cascaded channel, we derive a recursive \ac{mmse}
posterior update from scalar slot-wise observations and combine it
with Thompson sampling for uncertainty-aware RIS control. The posterior
recursion has per-slot complexity $\mathcal{O}(N^{2})$ and does not
grow with the slot index.
\item We analyze the Bayesian regret of the proposed \ac{mmse}-based Thompson-sampling
policy. A Bayesian regret decomposition is derived, showing that the
Bayesian received-power regret is controlled by the contraction of
the posterior covariance. We further establish a conditional sublinear
Bayesian-regret guarantee under a cumulative information growth condition.
\item We propose an \ac{sbl}-based Thompson-sampling algorithm that performs
posterior sampling in the sparse coefficient domain. A warm-started
hyperparameter refinement procedure is developed to exploit hybrid-field
sparsity while maintaining practical online complexity.
\item Simulations under \ac{los}-dominant, hybrid multipath, and Rayleigh
fading scenarios show that the proposed methods rapidly approach the
perfect-\ac{csi} benchmark in sparse channels and improve spectral
efficiency over estimation-based, random-probing, and finite-codebook
bandit baselines.
\end{itemize}

The rest of the paper is organized as follows. Section II presents
the system model and Bayesian bandit formulation. Section III develops
the \ac{mmse}-based Bayesian bandit beamforming scheme. Section IV
analyzes the regret performance. Section V introduces the hybrid-field
energy-focusing dictionary and the \ac{sbl}-based extension. Section
VI presents simulation results. Section VII concludes the paper.

\section{System Model and Bayesian Bandit Formulation}

In this section, we present the transmission model for a passive \ac{ris}-assisted
link under a single-pattern shared pilot-and-data architecture. We
then formulate the corresponding online \ac{ris} control problem.
Unlike conventional designs that separate training and transmission,
each \ac{ris} reflection pattern serves two roles in every slot.
It determines both the channel estimation quality and the data-transmission
performance.

\subsection{System Model}

\selectlanguage{american}%
\begin{figure}
\begin{centering}
\includegraphics[width=0.8\columnwidth]{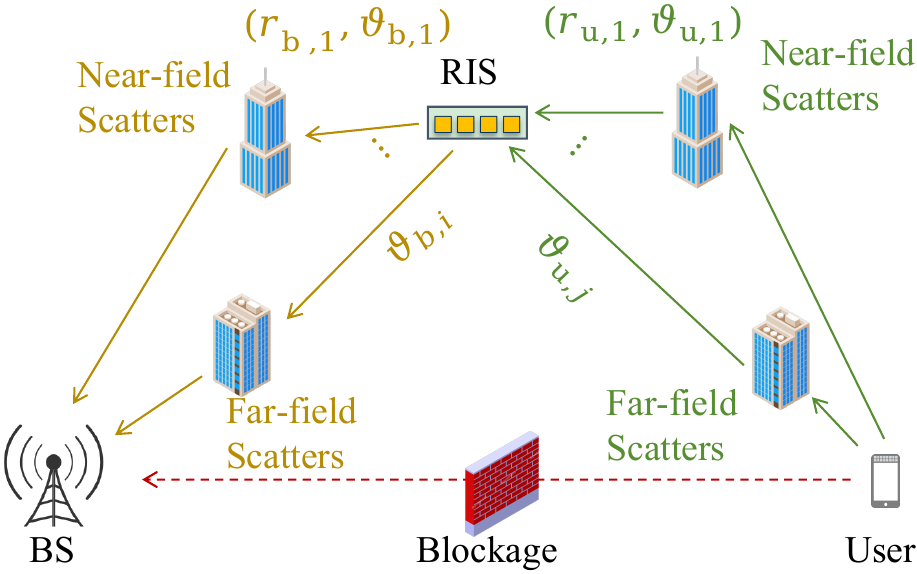}
\par\end{centering}
\caption{\protect\label{fig:capacity-1-2-2}Illustration of the BS-RIS-user
channel, where the direct path between the BS and the user is blocked.
Scatterers may appear in both the near-field and the far-field of
the RIS.}
\end{figure}

\selectlanguage{english}%
As shown in Fig.~\ref{fig:capacity-1-2-2}, we consider a single
user narrowband downlink system assisted by an \ac{ris}, where both
the \ac{bs} and the user are each equipped with a single antenna.
For clarity, we consider an \ac{ris} with $N$ reflecting elements
arranged as a \ac{ula}. The same principle can be extended to a UPA
by using two angular variables and a range variable. We assume that
the direct link between the \ac{bs} and the user is obstructed and
negligible. Therefore, the goal is to optimize the configuration of
the \ac{ris} to boost the \ac{bs}-\ac{ris}-user channel.

Let $\mathbf{h}_{\text{b}}=\left[h_{\text{b},1},h_{\text{b},2},\ldots,h_{\text{b},N}\right]^{\text{\ensuremath{\top}}}\in\mathbb{C}^{N}$
and $\mathbf{h}_{\text{u}}=\left[h_{\text{u},1},h_{\text{u},2},\ldots,h_{\text{u},N}\right]^{\text{\ensuremath{\top}}}\in\mathbb{C}^{N}$
denote the BS-RIS and RIS-user channel vectors, respectively. For
downlink transmission, define $h_{n}=h_{\text{b},n}h_{\text{u},n}$
as the equivalent cascaded coefficient through the RIS element $n$.
At slot $t$, the RIS applies the reflection vector $\mathbf{b}_{t}=[b_{t,1},b_{t,2},\ldots,b_{t,N}]^{\mathrm{\top}}\in\mathcal{B}$,
where $\mathcal{B}\triangleq\{\mathbf{b}\in\mathbb{C}^{N}:\;|[\mathbf{b}]_{n}|=1,\ \forall n\}$.
The resulting equivalent BS-RIS-user scalar channel is $\sum_{n=1}^{N}h_{n}b_{t,n}$.

In each time slot, $N_{\text{s}}$ symbols are transmitted. For example,
in a multi-carrier system, such as \ac{ofdm}, that consists of $N_{\text{s}}$
subcarriers, the $N_{\text{s}}$ symbols are multiplexed in the frequency
domain. If the channel is under flat fading, only a few pilot symbols
$N_{\text{p}}$ are needed for estimating the equivalent scalar channel
for data demodulation, while the remaining symbols $N_{\text{d}}$
can be used for data payload with $N_{\text{s}}=N_{\text{p}}+N_{\text{d}}$.
We define the payload fraction as $\eta\triangleq N_{\text{d}}/N_{\text{s}}\in(0,1]$.
It is assumed that the narrowband channel is in block fading, where
the channel $h_{n}$ remains unchanged for $T$ consecutive time slots.
Denote $\mathbf{x}_{t}\in\mathbb{C}^{N_{\text{s}}}$ as the vector
of transmitted symbols in the time slot $t$. The received signal
vector $\mathbf{y}_{t}$ is thus given by
\begin{align}
\mathbf{y}_{t} & =\sum_{n=1}^{N}h_{n}b_{t,n}\mathbf{x}_{t}+\bm{\xi}_{t}\label{eq:received=000020one}
\end{align}
where $\bm{\xi}_{t}\sim\mathcal{CN}(\bm{0},\sigma^{2}\mathbf{I})$
models the \ac{awgn}. 

To optimize the each RIS phase shift $b_{t,n}$, a conventional approach
first estimates the channel coefficients $h_{n}$ for all $n=1,2,\dots,N$.
This would require sending pilot sequences over $N$ time slots, because
different RIS configurations need to be distributed over time. If
the channel has sparsity, it is possible to send fewer pilots \cite{RIS_CS1,RIS_CS2}.
This motivates an online design that extracts channel information
while the RIS is already being used for communication.

\subsection{Bandit Beamforming and Problem Formulation}

As the equivalent channel $\sum_{n=1}^{N}h_{n}b_{t,n}$ in (\ref{eq:received=000020one})
is a scalar for a given RIS configuration $\mathbf{b}_{t}$, one or
a few pilot symbols are sufficient to estimate the scalar equivalent
channel $\sum_{n=1}^{N}h_{n}b_{t,n}$. The remaining symbols can be
used for payload transmission, so  the RIS configuration $\mathbf{b}_{t}$
should balance channel estimation and data transmission.

Such a philosophy falls into the paradigm of bandit approaches. The
\ac{bs} selects $\mathbf{b}_{t}$ adaptively from past observations.
Fig.~\ref{fig:capacity-1-2-1}(b) illustrates this online design
without a dedicated channel-estimation phase. By contrast, Fig.~\ref{fig:capacity-1-2-1}(a)
illustrates a conventional design with a dedicated channel-estimation
phase in the first $k$ time slots.

Here, we denote the channel $\mathbf{h}=\text{diag}(\mathbf{h}_{\text{b}})\mathbf{h}_{\text{u}}\in\mathbb{C}^{N}$
as the cascaded channel in vector form. Within one coherence block,
the cascaded channel is assumed unchanged over all $T$ slots. In
each slot, the same RIS reflection vector is used by the pilot symbols
and the data symbols. After normalization by the known pilot symbol,
one scalar pilot observation can be written as
\begin{equation}
y_{t}=\mathbf{b}_{t}^{\top}\mathbf{h}+\xi_{t}.\label{eq:scalar_pilot_observation}
\end{equation}
From the perspective of the \ac{bs} before receiving pilots, the
realization of $\mathbf{h}$ is assumed unknown. To enable the \ac{ris}
phase configuration, we represent the uncertainty of $\mathbf{h}$
through a statistical prior distribution $\mathbf{h}\sim\Pi_{0}$.
Let $\mathcal{F}_{t}\triangleq\sigma(\mathbf{b}_{1},y_{1},\mathbf{b}_{2},y_{2},\dots,\mathbf{b}_{t},y_{t})$
denote the information available after slot $t.$ Then, the corresponding
posterior belief of the channel is $\Pi_{t}\triangleq\mathcal{L}(\mathbf{h}|\mathcal{F}_{t})$
where $\mathcal{L}(\cdot|\cdot)$ is the conditional distribution.

We formulate bandit beamforming as a decision process that maps the
set of received signals $\{y_{1},y_{2},\dots,y_{t-1}\}$ to the \ac{ris}
configuration $\mathbf{b}_{t}$. Specifically, the \ac{ris} control
policy over the horizon $T$ is defined as a sequence of mapping $\mathbf{b}_{t}=\kappa(\Pi_{t-1})\in\mathcal{B},$
where $\kappa$ maps the currently available channel belief into a
feasible \ac{ris} configuration in $\mathcal{B}$. The goal of the
\ac{ris} controller is to find a beamforming strategy $\kappa$ to
maximize the cumulative expected effective rate
\begin{align}
\mathscr{P}:\quad\underset{\kappa}{\text{max}}\; & \mathbb{E}[\sum_{t=1}^{T}U_{t}(\mathbf{b}_{t};\mathbf{h})]\label{eq:problem}
\end{align}
where $b_{t}=\kappa(\Pi_{t-1})$ and the expectation is taken with
respect to the prior $\mathbf{h}\sim\Pi_{0}$, the noise process $\{\xi_{t}\}_{t=1}^{T}$,
and the internal randomness of the policy. The transmission utility
$U_{t}(\mathbf{b}_{t};\mathbf{h})$ is defined as the effective rate
as
\begin{equation}
U_{t}(\mathbf{b}_{t};\mathbf{h})=\eta\log_{2}(1+\frac{1}{\sigma^{2}}|\mathbf{b}_{t}^{\text{\ensuremath{\top}}}\mathbf{h}|^{2}).\label{eq:effective_rate}
\end{equation}

\selectlanguage{american}%
\begin{figure}
\begin{centering}
\includegraphics[width=1\columnwidth]{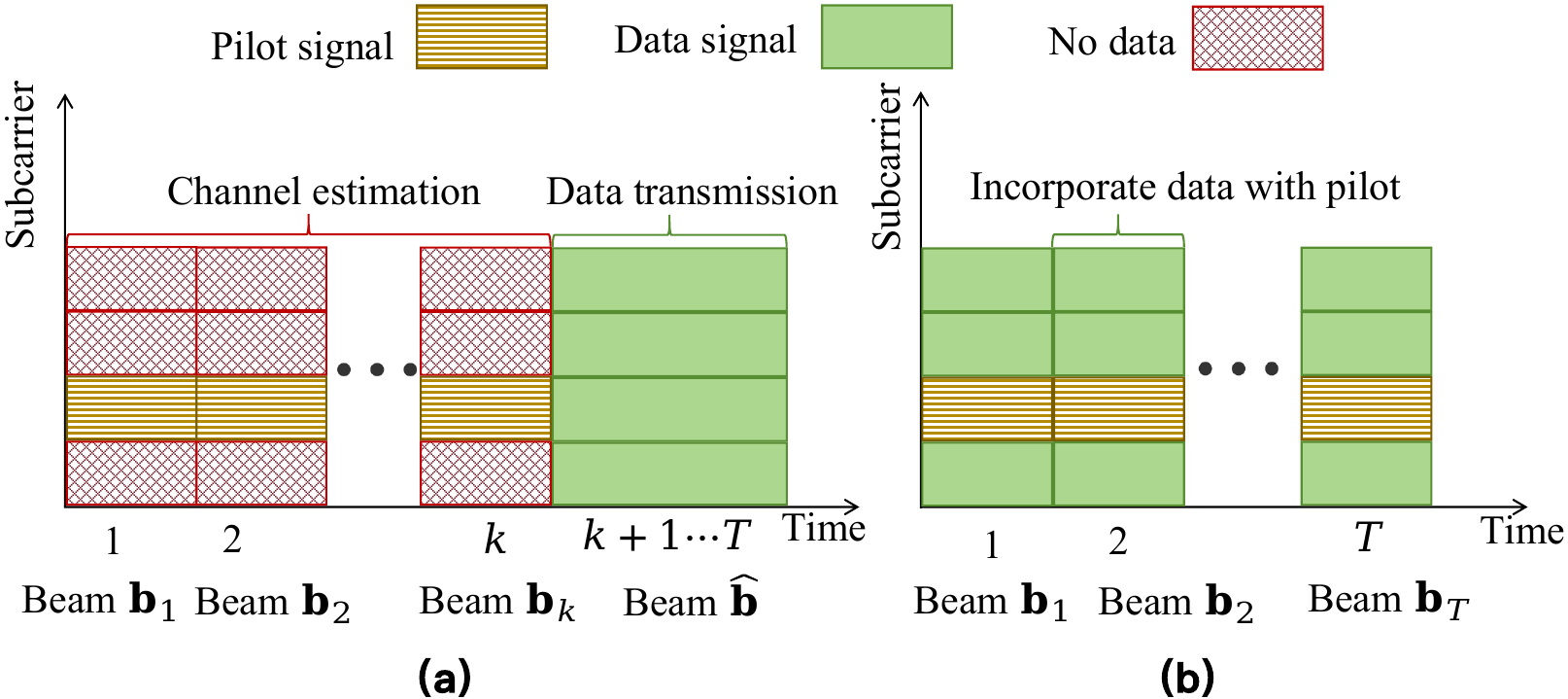}
\par\end{centering}
\caption{\protect\label{fig:capacity-1-2-1}Illustration of the frame structure
for CSI acquisition and data transmission. (a) The conventional channel
training and data transmission scheme (b) The proposed online bandit-beamforming
strategy embeds pilot symbols under each adaptive RIS configuration
while using the remaining resources for payload transmission.}
\end{figure}

\selectlanguage{english}%

\section{Bayesian Bandit Beamforming with\\ Implicit Channel Learning\protect\label{sec:Bayesian-Bandit-Beamforming}}

For notational simplicity, we first consider that only one symbol
in $\mathbf{x}_{t}$ is used as a pilot symbol, and assume this pilot
symbol is $x_{t,1}=1$. As we will only focus on the received signal
$y_{t}$ of this pilot symbol $x_{t,1}=1$, and simplify our notation
from the vector form in (\ref{eq:received=000020one}) to a scalar
form $y_{t}=\sum_{n=1}^{N}h_{n}b_{t,n}+\xi_{t}$. If multiple pilot
symbols are used under the same RIS pattern, their observations can
be averaged into the same scalar model with an effective noise variance. 

Conventional Bayesian bandit usually makes {\em discrete} decisions
from a finite set of possible actions. In the context of RIS beamforming,
it corresponds to drawing the RIS configuration $\mathbf{b}_{t}$
from a dictionary that contains finite elements. When the dictionary
is large, it is challenging to establish Bayesian posterior for each
of the elements in the dictionary, whereas, when the dictionary is
small, the RIS beamforming is strictly suboptimal. Hence, we propose
to employ a Bayesian model for the channel in a {\em continuous}
space and update its posterior in an iterative way for the improved
channel information.

\subsection{Implicit Channel Learning and Bayesian Posterior}

To solve problem (\ref{eq:problem}), we first assume that the equivalent
channel $\mathbf{h}$ follows the prior $\mathcal{CN}(\mathbf{0},\bm{\Sigma}_{0})$
at the initial time $t=0$. We assume that $\bm{\Sigma}_{0}$ is Hermitian
positive definite. If no prior information about the channel is available,
we simply set $\bm{\Sigma}_{0}=\mathbf{I}$. Given this Gaussian prior,
we update the channel belief and obtain the Bayesian posterior $\mathcal{CN}(\bm{\mu}_{t},\bm{\Sigma}_{t})$
based on the observations $\{y_{1},y_{2},\dots,y_{t}\}$ under the
RIS configuration $\{\mathbf{b}_{1},\mathbf{b}_{2},\dots,\mathbf{b}_{t}\}$
in the first $t$ time slots as follows. 

\subsubsection{Bayesian Channel Posterior}

Define the observation vector up to slot $t$ as $\mathbf{y}_{1:t}=[y_{1},y_{2},\dots,y_{t}]^{\text{\ensuremath{\top}}}$.
We also collect the corresponding RIS reflection vectors in $\mathbf{B}_{1:t}=[\mathbf{b}_{1},\mathbf{b}_{2},\dots,\mathbf{b}_{t}]\in\mathbb{C}^{N\times t}$.
The stacked received signal is
\begin{equation}
\mathbf{y}_{1:t}=\mathbf{B}_{1:t}^{\text{\ensuremath{\top}}}\mathbf{h}+\bm{\xi}_{1:t}\label{eq:received=000020signal}
\end{equation}
where the vector $\bm{\xi}_{1:t}=[\xi_{1},\xi_{2},\ldots,\xi_{t}]^{\text{\ensuremath{\top}}}$
denotes the aggregated noise vector up to time $t$ and $\xi_{t}\sim\mathcal{CN}(0,\sigma^{2})$. 

Under the linear observation model (\ref{eq:received=000020signal})
and the Gaussian prior for $\mathbf{h}$, an \ac{mmse} estimator
$\hat{\mathbf{h}}_{t}$ that minimizes the error $\mathbb{E}[||\hat{\mathbf{h}}_{t}-\mathbf{h}||^{2}]$
is given by 
\begin{equation}
\hat{\mathbf{h}}_{t}=\frac{1}{\sigma^{2}}(\bm{\Sigma}_{0}^{-1}+\frac{1}{\sigma^{2}}\mathbf{B}_{1:t}^{*}\mathbf{B}_{1:t}^{\top})^{-1}\mathbf{B}_{1:t}^{*}\mathbf{y}_{1:t}.\label{eq:posterior=000020mut}
\end{equation}
In addition, the error covariance matrix $\mathbb{E}[(\hat{\mathbf{h}}_{t}-\mathbf{h})(\hat{\mathbf{h}}_{t}-\mathbf{h})^{\text{H}}]$
is given as
\begin{align}
\bm{\Sigma}_{t} & =(\bm{\Sigma}_{0}^{-1}+\frac{1}{\sigma^{2}}\mathbf{B}_{1:t}^{*}\mathbf{B}_{1:t}^{\top})^{-1}\nonumber \\
 & =\bm{\Sigma}_{0}-\bm{\Sigma}_{0}\mathbf{B}_{1:t}^{*}(\sigma^{2}\mathbf{I}_{t}+\mathbf{B}_{1:t}^{\top}\bm{\Sigma}_{0}\mathbf{B}_{1:t}^{*})^{-1}\mathbf{B}_{1:t}^{\top}\bm{\Sigma}_{0}\label{eq:posterior=000020sigma}
\end{align}
where the equation (\ref{eq:posterior=000020sigma}) is obtained via
the Woodbury matrix identity \cite{kay1993fundamentals}.

While the estimator is given by closed-form expression (\ref{eq:posterior=000020mut}),
it is important to note that the dimension of the matrix $\mathbf{B}_{1:t}$
increases over time, since a new phase shift vector $\mathbf{b}_{t}$
is appended to $\mathbf{B}_{1:t-1}$ at each time step $t$. As a
result, computing the estimate $\hat{\mathbf{h}}_{t}$ in (\ref{eq:posterior=000020mut})
has an increasing complexity as $t$ evolves.

\subsubsection{Posterior Update from Sequential Observations}

Note that the posterior (\ref{eq:posterior=000020mut}) and (\ref{eq:posterior=000020sigma})
can be computed in an iterative way. Consider a complex Gaussian prior
$\mathcal{CN}(\bm{\mu}_{t-1},\bm{\Sigma}_{t-1})$ after the Bayesian
channel estimation at time slot $t-1$. Recall from (\ref{eq:received=000020one})
that the observation in time slot $t$ can be written in vector form
as $y_{t}=\mathbf{b}_{t}^{\top}\mathbf{h}+\xi_{t}.$

Following the same result in \cite[Theorem 12.1]{kay1993fundamentals},
the \ac{mmse} estimator is given by
\begin{equation}
\hat{\mathbf{h}}_{t}=\bm{\mu}_{t-1}+\sigma^{-2}\bm{\Sigma}_{t}\mathbf{b}_{t}^{*}(y_{t}-\mathbf{b}_{t}^{\top}\bm{\mu}_{t-1})\label{eq:iterative=000020mu}
\end{equation}
and the posterior error covariance matrix is updated as
\begin{align}
\bm{\Sigma}_{t} & =\left(\bm{\Sigma}_{t-1}^{-1}+\sigma^{-2}\mathbf{b}_{t}^{*}\mathbf{b}_{t}^{\top}\right)^{-1}\nonumber \\
 & =\bm{\Sigma}_{t-1}-\bm{\Sigma}_{t-1}\mathbf{b}_{t}^{*}(\sigma^{2}+\mathbf{b}_{t}^{\top}\bm{\Sigma}_{t-1}\mathbf{b}_{t}^{*})^{-1}\mathbf{b}_{t}^{\top}\bm{\Sigma}_{t-1}\label{eq:iterative=000020sigma}
\end{align}
where (\ref{eq:iterative=000020sigma}) is obtained via the Sherman-Morrison
formula \cite{horn2012matrix}.

Note that as a key property of the \ac{mmse} estimator, the posterior
mean $\bm{\mu}_{t}\triangleq\mathbb{E}[\mathbf{h}|\mathcal{F}_{t}]$
is equal to the MMSE estimate $\hat{\mathbf{h}}_{t}$$.$ Such a posterior
becomes the prior for the next time slot, and it follows that
\begin{equation}
\bm{\mu}_{t}=\bm{\mu}_{t-1}+\sigma^{-2}\bm{\Sigma}_{t}\mathbf{b}_{t}^{*}(y_{t}-\mathbf{b}_{t}^{\top}\bm{\mu}_{t-1}).
\end{equation}
One can show that the estimate (\ref{eq:iterative=000020mu}) is equivalent
to (\ref{eq:posterior=000020mut}), but has a lower computational
complexity; similarly, (\ref{eq:iterative=000020sigma}) is equivalent
to (\ref{eq:posterior=000020sigma}). Specifically, using the Woodbury
form, the computational complexity of the batch MMSE estimate (\ref{eq:posterior=000020mut})
is $\mathcal{O}(N^{2}t+Nt^{2}+t^{3})$, depending on the implementation,
where the complexity quickly scales up as $t$ increases. By contrast,
the iterative approach (\ref{eq:iterative=000020mu}) has a complexity
of $\mathcal{O}(N^{2})$ which does not scale with $t$. This $\mathcal{O}(N^{2})$
count covers the posterior mean and covariance recursion. 

\subsection{Thompson Sampling for RIS Configuration \protect\label{subsec:Thompson-Sampling-for-RIS-configuration}}

Recall that the continuous action space for the \ac{ris} configuration
as $\mathcal{B}$. Given the Bayesian channel posterior $\mathbf{h}\sim\mathcal{CN}(\bm{\mu}_{t-1},\bm{\Sigma}_{t-1})$
obtained from the $t-1$ observations and since the effective-rate
utility is monotone increasing in the received power, we equivalently
use the slot-wise reward $r_{t}=|\mathbf{b}_{t}^{\top}\mathbf{h}|^{2}$
in (\ref{eq:effective_rate}), a greedy benchmark would use the posterior
mean estimate $\hat{\mathbf{h}}=\bm{\mu}_{t-1}$ of $\mathbf{h}$
and then choose a $\mathbf{b}\in\mathcal{U}$ that maximizes the reward
$|\mathbf{b}^{\top}\bm{\mu}_{t-1}|^{2}$.

For a fixed posterior mean, this greedy phase-alignment problem has
a closed-form global solution. Its weakness is not local optimality,
but premature exploitation. If the posterior mean is inaccurate in
early slots, repeatedly aligning to it may fail to probe high-uncertainty
channel directions and may slow posterior contraction. Thompson sampling
provides a Bayesian mechanism to balance channel learning (exploration)
and data transmission (exploitation). Specifically, it samples from
the distribution of the reward $r_{t}$ according to the posterior
of $\mathbf{h}$. Towards this end, at the beginning of the time slot
$t$, we first obtain a sample $\mathbf{h}'_{t}$ from the distribution
\begin{equation}
\mathbf{h}'_{t}\sim\mathcal{CN}(\bm{\mu}_{t-1},\bm{\Sigma}_{t-1}).
\end{equation}
Then, the RIS configuration as the action at the time slot $t$ is
obtained as the solution to the following problem
\begin{alignat}{1}
\mathop{\mbox{maximize}}\limits_{\mathbf{b}\in\mathcal{B}} & \quad\left|\mathbf{b}^{\top}\mathbf{h}'_{t}\right|^{2}.\label{eq:11-3}
\end{alignat}

Denote the sample $\mathbf{h}_{t}'=(h_{t,1}',h_{t,2}',\dots,h_{t,N}')$
as $h'_{t,n}=\alpha'_{n}e^{j\theta'_{n}}$. It is known from the literature
that the optimal solution to (\ref{eq:11-3}) is given as $b_{n}=e^{-j\theta'_{n}}$
(See Equation (21) in \cite{phaseshift}).

\begin{algorithm}[t]
\caption{Recursive MMSE-based Thompson sampling}
\label{alg:mmse-ts}
\begin{algorithmic}[1]
\REQUIRE Prior $\boldsymbol\mu_0=\mathbf 0$, covariance $\boldsymbol\Sigma_0$, noise variance $\sigma^2$, and horizon $T$.
\FOR{$t=1,2,\ldots,T$}
\STATE Draw $\mathbf h'_t\sim\mathcal{CN}(\boldsymbol\mu_{t-1},\boldsymbol\Sigma_{t-1})$.
\STATE Set $[\mathbf b_t]_n=\exp\{-j\angle([\mathbf h'_t]_n)\}$ for $n=1,\ldots,N$.
\STATE Use $\mathbf b_t$ for the pilot and payload symbols in slot $t$ and observe $y_t$.
\STATE Update $\boldsymbol\mu_t$ and $\boldsymbol\Sigma_t$ by \eqref{eq:iterative=000020mu} and \eqref{eq:iterative=000020sigma}.
\ENDFOR
\ENSURE RIS configurations $\{\mathbf b_t\}$ and posterior sequence $\{\boldsymbol\mu_t,\boldsymbol\Sigma_t\}$.
\end{algorithmic}
\end{algorithm}

\section{Performance Analysis}

In this section, we analyze the long-term performance of the proposed
Bayesian bandit beamforming policy. Since the effective-rate utility
(\ref{eq:effective_rate}) is a monotone and Lipschitz-continuous
function of the received power, we first study the Bayesian regret
in terms of the received-power regret and then relate it to the effective-rate
regret.

\subsection{Bayesian Regret}

Recall that the slot-wise communication utility is defined by the
effective rate in (\ref{eq:effective_rate}), which is an increasing
function of the received signal power $|\mathbf{b}_{t}^{\top}{\bf h}|^{2}$.
To make the regret analysis precise, define the scalar rate function
$u(x)\triangleq\eta\log_{2}(1+\sigma^{-2}x)$ for received power $x\ge0$.
Then $U_{t}(\mathbf{b};\mathbf{h})=u(|\mathbf{b}^{\top}\mathbf{h}|^{2})$.
Since $u(x)$ is increasing, the action that maximizes received power
also maximizes effective rate for a fixed channel realization. We
therefore first analyze regret in the received-power domain and then
translate it to effective-rate regret using the Lipschitz continuity
of $u(x)$. Specifically
\begin{equation}
u'(x)=\frac{\eta\sigma^{-2}}{\ln2}\frac{1}{1+\sigma^{-2}x}\le\frac{\eta}{\sigma^{2}\ln2},\qquad x\ge0
\end{equation}
so $u(x)$ is Lipschitz continuous \cite{real_analysis}. Accordingly,
we define the Bayesian received-power regret as
\begin{equation}
\mathcal{R}_{T}\triangleq\mathbb{E}\left[\sum_{t=1}^{T}\Big(|\mathbf{b}_{*}^{\text{\ensuremath{\top}}}\mathbf{h}|^{2}-|\mathbf{b}_{t}^{\text{\ensuremath{\top}}}\mathbf{h}|^{2}\Big)\right]\label{eq:bayes_snr_regret_sec2}
\end{equation}
where $\mathbf{b}_{*}\in\arg\max_{\mathbf{b}\in\mathcal{B}}|\mathbf{b}^{\top}\mathbf{h}|^{2}$
denotes an optimal RIS reflection vector under the true channel realization
$\mathbf{h}$. Thus, $\left|\mathbf{b}_{*}^{\top}\mathbf{h}\right|^{2}$
is the oracle received power for the channel realization $\mathbf{h}$.
If $\mathcal{R}_{T}=o(T)$, then the time-averaged expected received-power
regret $\mathcal{R}_{T}/T$ vanishes. This makes received-power regret
a convenient surrogate for proving long-term learning performance.

The corresponding Bayesian effective-rate regret is
\begin{equation}
\mathcal{R}_{T}^{{\rm rate}}\triangleq\mathbb{E}\left[\sum_{t=1}^{T}\big(U_{t}(\mathbf{b}_{*};\mathbf{h})-U_{t}(\mathbf{b}_{t};\mathbf{h})\big)\right].
\end{equation}
By the Lipschitz bound above, it satisfies
\begin{equation}
\mathcal{R}_{T}^{{\rm rate}}\le\frac{\eta}{\sigma^{2}\ln2}\mathcal{R}_{T}.\label{eq:rate_power_regret_relation}
\end{equation}
Therefore, sublinear Bayesian received-power regret also implies sublinear
Bayesian effective-rate regret.

\subsection{Regret Decomposition of the Proposed Algorithm }

We first quantify the one-slot regret caused by aligning the RIS phases
to a posterior sample rather than to the true cascaded channel. The
following lemma upper-bounds this regret in terms of the mismatch
between the sampled channel and the true channel.
\begin{lem}[\textbf{One-slot Mismatch Bound}]
\label{lem:One-slot-Mismatching-Bound} For any $\mathbf{h},\mathbf{h}'_{t}\in\mathbb{C}^{N}$,
let $\mathbf{b}_{*}$ and $\mathbf{b}_{t}$ be the unit-modulus phase-alignment
vectors for $\mathbf{h}$ and $\mathbf{h}'_{t}$, respectively, i.e.,
$[\mathbf{b}_{*}]_{n}=e^{-j\angle(h_{n})}$ and $[\mathbf{b}_{t}]_{n}=e^{-j\angle(h'_{t,n})}$.
Then
\begin{equation}
|\mathbf{b}_{*}^{\top}\mathbf{h}|^{2}-|\mathbf{b}_{t}^{\text{\ensuremath{\top}}}\mathbf{h}|^{2}\leq4N||\mathbf{h}||_{2}||\mathbf{h}-\mathbf{h}'_{t}||_{2}.\label{eq:one_slot_decomposition}
\end{equation}
\end{lem}
\begin{proof}
See Appendix \ref{sec:proof_one_slot_mismatching}.
\end{proof}
Lemma \ref{lem:One-slot-Mismatching-Bound} characterizes instantaneous
received-power regret caused by aligning \ac{ris} phases to a sampled
channel $\mathbf{h}'_{t}$. This shows that in each iteration, the
regret can be controlled by the channel mismatch $||\mathbf{h}-\mathbf{h}'_{t}||_{2}$.
Combined with the posterior-sampling property and posterior covariance
contraction, this deterministic mismatch bound leads to the Bayesian
regret decomposition below.
\begin{thm}[\textbf{Bayesian Regret Decomposition}]
\label{thm:Bayesian-regret-decomposition} Under the well-specified
Bayesian observation model in Sections II-III, assume that $S^{2}\triangleq\mathbb{E}[||\mathbf{h}||_{2}^{2}]<\infty$.
The Bayesian regret for the proposed algorithm satisfies
\begin{equation}
\mathcal{R}_{T}\leq4\sqrt{2}NS\sum_{t=1}^{T}\sqrt{\mathbb{E}[\text{tr}(\bm{\Sigma}_{t-1})]}\label{eq:Bayesian_regret_decomposition}
\end{equation}
where the expectation is taken over the prior, the noise, and the
internal randomness of the Thompson sampling.
\end{thm}
\begin{proof}
See Appendix \ref{sec:proof_bayesian_regret_decomposition}.
\end{proof}
Theorem \ref{thm:Bayesian-regret-decomposition} shows that the Bayesian
received-power regret of the proposed algorithm is governed by the
contraction of the posterior uncertainty of the effective cascaded
channel. Therefore, the regret analysis reduces to understanding how
fast $\mathbb{E}[\text{tr}(\bm{\Sigma}_{t-1})]$ decreases over time.
We next give explicit conditions under which this contraction is fast
enough to guarantee sublinear Bayesian regret.

\subsection{Condition for Sublinear Bayesian Regret}

We next convert the decomposition (\ref{eq:Bayesian_regret_decomposition})
into explicit sufficient conditions for sublinear Bayesian regret
\begin{cor}[\textbf{Sublinear Bayesian Regret Under Sufficient Information Accumulation}]
\label{cor:Sublinear=000020Bayesian=000020regret} Assume the well-specified
Bayesian observation model and the channel follows Rayleigh fading,
$\mathbf{h}\sim\mathcal{CN}(\mathbf{0},\mathbf{I}_{N})$, so that
$S^{2}\triangleq\mathbb{E}[\|\mathbf{h}\|_{2}^{2}]=N$. Define $\mathbf{G}_{t}\triangleq\sum_{s=1}^{t}\mathbf{b}_{s}^{*}\mathbf{b}_{s}^{\top}$
and suppose that the posterior covariance satisfies $\bm{\Sigma}_{t}=(\mathbf{I}_{N}+\sigma^{-2}\mathbf{G}_{t})^{-1}$.
Let $t_{0}\geq1$ and define $\bar{t}_{0}\triangleq\max\{t_{0},N\}$.
If there exists a constant $C>0$ such that, almost surely for the
realized RIS sequence
\begin{equation}
\lambda_{\text{\ensuremath{\min}}}(\mathbf{G}_{t})\geq C(t-\bar{t}_{0}+1),\hspace*{1em}\forall t\geq\bar{t}_{0}.\label{eq:condition_corollary_3}
\end{equation}
Then,
\begin{equation}
\mathcal{R}_{T}=O(\sqrt{T}),\hspace*{1em}T\rightarrow\infty.
\end{equation}

In particular, when $t_{0}\le N$, the condition reduces to $\lambda_{\min}(\mathbf{G}_{t})\ge C(t-N+1)$
and the first $N$ slots contribute only a finite regret term as $4\sqrt{2}SN^{5/2}$,
equivalently $4\sqrt{2}N^{3}$ under the Rayleigh prior, so that the
$O(\sqrt{T})$ bound remains valid.
\end{cor}
\begin{proof}
See Appendix \ref{sec:proof_Sublinear=000020Bayesian=000020regret}.
\end{proof}
Corollary \ref{cor:Sublinear=000020Bayesian=000020regret} shows that
under isotropic Rayleigh fading, sublinear Bayesian regret is guaranteed
if the realized RIS configurations remain informative across all channel
dimensions after a finite initial period. However, in our \ac{ris}
system, the phase configuration simultaneously serves channel estimation
and data transmission, making it difficult to enforce strong exploration
in every slot. Motivated by this, we consider that only require the
accumulated RIS configurations to provide sufficient information over
time and gradually reduce the posterior uncertainty of the cascaded
channel. We next show that this cumulative information requirement
is sufficient for sublinear Bayesian regret and derive a simple sufficient
condition.

To express this information-growth requirement, define the conditional
expected information matrix as $\mathbf{Q}_{t}\triangleq\mathbb{E}[\mathbf{b}_{t}^{*}\mathbf{b}_{t}^{\top}\mid\mathcal{F}_{t-1}]$,
where $\mathbf{b}_{t}^{*}\mathbf{b}_{t}^{\top}$ is the information
contribution of the RIS configuration to the posterior precision matrix
in slot $t$.
\begin{thm}[\textbf{Conditional Sublinear Bayesian Regret}]
 \label{thm:sublinear=000020comm=000020with=000020channel=000020excitation}
Assume that the Bayesian observation model is well specified with
$\mathbf{h}\sim\mathcal{CN}(\mathbf{0},\mathbf{I}_{N})$, and let
$\{\mathbf{b}_{t}\}_{t=1}^{T}$ be generated by the proposed algorithm.
Suppose that there exist constants $c_{\alpha}>0$ and $\alpha\in(1/2,1]$,
and a nonnegative sequence $\{q_{t,\alpha}\}_{t\ge1}$, such that
\begin{equation}
\mathbb{P}\left(\frac{\lambda_{\min}\left(\sum_{s=1}^{t}\mathbf{Q}_{s}\right)}{t^{\alpha}}\le c_{\alpha}\right)\le q_{t,\alpha},\qquad\forall t\ge1.\label{eq:excitation=000020condition}
\end{equation}
Then the cumulative Bayesian received-power regret satisfies
\begin{align}
\mathcal{R}_{T}\le & 4\sqrt{2}N^{3/2}\Bigg(\sqrt{N}+\sum_{t=1}^{T-1}\Bigg(\frac{N}{1+\frac{c_{\alpha}}{2\sigma^{2}}t^{\alpha}}+Nq_{t,\alpha}\nonumber \\
 & \hspace*{1em}\hspace*{1em}+N^{2}\exp(-\frac{c_{\alpha}^{2}}{128N^{2}}t^{2\alpha-1})\Bigg)^{1/2}\Bigg).\label{eq:cumulative_regret_bound_perf_main}
\end{align}
Consequently, if $\sum_{t=1}^{T}\sqrt{q_{t,\alpha}}=o(T)$, then $\mathcal{R}_{T}/T\to0$
as $T\to\infty$.
\end{thm}
\begin{proof}
See Appendix \ref{sec:proof=000020sublinear=000020comm=000020with=000020channel=000020excitation}.
\end{proof}
Theorem \ref{thm:sublinear=000020comm=000020with=000020channel=000020excitation}
gives a cumulative information-growth condition under which sublinear
Bayesian regret follows. Specifically, it requires the minimum eigenvalue
of the conditional expected information accumulated by the RIS configurations
to grow at least polynomially with exponent $\alpha>1/2$. Because
this property is determined by the policy-induced posterior distribution,
which has no closed-form expression, we evaluate it through a finite-horizon
simulation.

Under the Rayleigh fading setting $\mathbf{h}\sim\mathcal{CN}(\mathbf{0},\mathbf{I}_{N})$,
we run the MMSE-based Thompson-sampling algorithm. Since $\mathbf{Q}_{t}$
is a conditional expectation over the posterior-sampling randomness
given $\mathcal{F}_{t-1}$, we approximate it along each adaptive
trajectory by the Monte Carlo average $\widehat{\mathbf{Q}}_{t}=\frac{1}{\Upsilon}\sum_{\ell=1}^{\Upsilon}\mathbf{b}_{t,\ell}^{*}\mathbf{b}_{t,\ell}^{\top}$,
where $\Upsilon$ is the number of Monte Carlo samples. Each $\mathbf{b}_{t,\ell}$
is obtained by drawing an independent posterior sample from $\mathcal{CN}(\bm{\mu}_{t-1},\bm{\Sigma}_{t-1})$
and then applying the phase-alignment rule in (\ref{eq:11-3}). The
estimate $\widehat{\mathbf{Q}}_{t}$ is used only to diagnose cumulative
information growth along simulated horizon.

\selectlanguage{american}%
\begin{figure}
\begin{centering}
\includegraphics[width=0.8\columnwidth]{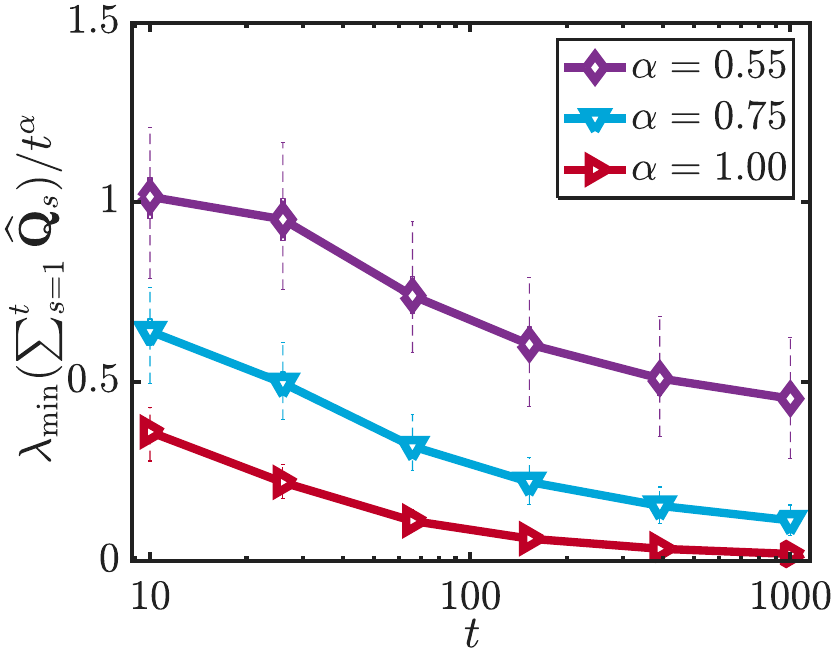}
\par\end{centering}
\caption{\protect\label{fig:tho2_validation_1}\foreignlanguage{english}{Cumulative
information growth for the Thompson sampling policy under Rayleigh
fading. The results provide finite-horizon numerical evidence that
the cumulative information measure remains positive for choices $(\alpha,c_{\alpha})=(0.55,0.2),(0.75,0.1),(1,0.02)$,
supporting the sufficient condition in Theorem \ref{thm:sublinear=000020comm=000020with=000020channel=000020excitation}.}}
\end{figure}

\selectlanguage{english}%
Fig. \ref{fig:tho2_validation_1} reports the cumulative information
growth of the Thompson-sampling policy under Rayleigh fading over
$t=[10,\dots,1000]$. For each exponent $\alpha\in\{0.55,0.75,1\}$,
we plot the normalized cumulative information measure $\widehat{\Lambda}_{t}/t^{\alpha}$
where $\widehat{\Lambda}_{t}\triangleq\lambda_{\min}\left(\sum_{s=1}^{t}\widehat{\mathbf{Q}}_{s}\right)$.
The results show that cumulative information remains above positive
constants $c_{\alpha}$ for the considered pairs $\ensuremath{(\alpha,c_{\alpha})=(0.55,0.2)},\ensuremath{(0.75,0.1)},$
and $\ensuremath{(1,0.02)}.$ Thus, this experiment indicates that
the Thompson-sampling policy can generate the cumulative information
growth required by $(\ref{eq:excitation=000020condition})$ over the
horizon, thereby supporting the sufficient condition for sublinear
Bayesian regret.

\section{Bayesian Bandit Exploiting Sparsity in Hybrid Near/Far-field Propagation}

In this section, we extend the posterior-driven RIS control framework
to exploit sparsity in hybrid near-/far-field propagation. The extension
preserves the Bayesian bandit formulation, but represents the cascaded
channel over a structured angle-distance dictionary. Thompson sampling
is then performed in the coefficient domain, with sparse Bayesian
learning used to maintain a sparsity-aware posterior. The main design
challenge is to construct a dictionary that captures hybrid-field
steering responses while remaining compact enough for efficient posterior
updates and bandit beamforming.

Although several dictionary design methods exist in the literature,
such as \cite{columncovariance,3dB_coherence}, these dictionaries
are {\em not} specifically designed to accelerate bandit beamforming
with sparse channels. To fill this gap, this section considers a general
hybrid near/far-field propagation scenario, where some scatterers
may lie in the near field of the RIS and others may lie in the far
field. Because the bandit reward is the effective rate in (\ref{eq:effective_rate}),
we propose an {\em energy focusing} criterion for the dictionary
design. 

\subsection{The Hybrid Near/Far-field Channel}

\selectlanguage{american}%
This subsection focuses on the scenario where the cascaded channel
is dominated by a small number of strong propagation paths. Recall
the\foreignlanguage{english}{ \ac{ris} model under the \ac{ula}
setting with $N$ reflecting elements spaced by $d=\lambda/2$ for
simplicity. For a scatterer at distance $r$ and physical angle $\theta$,
we define the normalized angular parameter as $\vartheta=\sin\theta\in[-1,1]$.
The near-field steering vector is given by \cite{columncovariance}
\begin{equation}
\mathbf{s}(\vartheta,r)=\frac{1}{\sqrt{N}}[e^{j\pi\phi_{0}(\vartheta,r)},e^{j\pi\phi_{1}(\vartheta,r)},...,e^{j\pi\phi_{N-1}(\vartheta,r)}]^{\top}\label{eq:steering=000020vector}
\end{equation}
where
\begin{equation}
\phi_{n}(\vartheta,r)\triangleq-n\vartheta+\frac{n^{2}d}{2r}(1-\vartheta^{2}),\hspace*{1em}n=0,1,\dots,N-1.\label{eq:fresnel_approximation}
\end{equation}
}

\selectlanguage{english}%
While the above formula also applies to the far-field case, the steering
vector (\ref{eq:steering=000020vector}) reduces to
\begin{equation}
\mathbf{s}(\vartheta,r=\infty)=\frac{1}{\sqrt{N}}[1,e^{-j\pi\vartheta},...,e^{-j\pi(N-1)\vartheta}]^{\top}\label{eq:near_field=000020steering=000020vector}
\end{equation}
which is asymptotically accurate at large $r\gg Z\triangleq2N^{2}d^{2}/\lambda=N^{2}d$,
where $Z$ represents the Rayleigh distance.

Hence, the multipath equivalent concatenated channel from the BS to
the user via the $N$-element RIS is given by
\begin{equation}
\mathbf{h}=\Big(\sum_{i}\alpha_{i}\mathbf{s}(\vartheta_{\text{b},i},r_{\text{b,}i})\Big)\odot\Big(\sum_{j}\beta_{j}\mathbf{s}(\vartheta_{\text{u},j},r_{\text{u,}j})\Big)\label{eq:multipath-equiv-chann}
\end{equation}
where the term $\sum_{i}\alpha_{i}\mathbf{s}(\vartheta_{\text{b},i},r_{\text{b,}i})$
represents the narrowband multipath channel for the BS-RIS link with
$\alpha_{i}$ being the gain of the path $i$ from a scatter in $\vartheta_{\text{b},i}$
at distance $r_{\text{b,}i}$, the term $\sum_{j}\beta_{j}\mathbf{s}(\vartheta_{\text{u},j},r_{\text{u,}j})$
represents that for the RIS-user link, and the operator $\odot$ represents
the Hadamard product.

Let $\nu\triangleq\vartheta_{\mathrm{b}}+\vartheta_{\mathrm{u}}$
be the cascaded angular parameter. Although $\nu$ can lie in $[-2,2]$,
we have
\begin{equation}
\vartheta_{{\rm eq}}=\begin{cases}
\nu+2, & \nu<-1,\\
\nu, & -1\le\nu\le1,\\
\nu-2, & \nu>1
\end{cases}\label{eq:folding}
\end{equation}
that $\vartheta_{{\rm eq}}\in[-1,1]$ as the folded representative
of $\nu$ \cite{csris}. We show in the following result that the
multipath equivalent channel can be represented as a linear combination
of vectors of the form $\mathbf{s}(\vartheta,r)$ in (\ref{eq:steering=000020vector}).
\begin{prop}[\textbf{Unified Representation}]
 \label{prop:unified-dict-vector} The steering vector of the corresponding
equivalent cascaded RIS channel can be written as
\begin{equation}
\mathbf{s}(\vartheta_{\mathrm{b}},r_{\mathrm{b}})\odot\mathbf{s}(\vartheta_{\mathrm{u}},r_{\mathrm{u}})=\frac{1}{\sqrt{N}}\mathbf{s}(\vartheta_{{\rm eq}},r_{{\rm eq}})\label{eq:unified_half_wavelength_product}
\end{equation}
where the corresponding effective range is
\begin{equation}
r_{{\rm eq}}=\frac{1-\vartheta_{{\rm eq}}^{2}}{(1-\vartheta_{\mathrm{b}}^{2})/r_{\mathrm{b}}+(1-\vartheta_{\mathrm{u}}^{2})/r_{\mathrm{u}}}.\label{eq:req_positive_visible}
\end{equation}
If the denominator in (\ref{eq:req_positive_visible}) is zero, we
set $r_{{\rm eq}}=\infty$. Consequently, if $|\vartheta_{\text{eq}}|<1$,
the equivalent cascaded RIS channel can be represented as
\begin{equation}
\mathbf{h}=\sum_{k}\frac{c_{k}}{\sqrt{N}}\mathbf{s}(\vartheta_{k},r_{k})\label{eq:unified_sparse_cascaded_channel}
\end{equation}
where $k$ indexes those BS-RIS/RIS-user path pairs and $c_{k}$ absorbs
the corresponding composite path gain.
\end{prop}
\begin{proof}
See Appendix \ref{sec:unified-dict-vector-proof}.
\end{proof}

\subsection{Supporting Dictionary via Energy Focusing\protect\label{subsec:dictionary}}

\subsubsection{Motivation and Design Constraints}

For a large RIS system with many reflecting elements, the hybrid-field
channel is usually dominated by a small number of scatterers, so the
representation in (\ref{eq:multipath-equiv-chann}) contains only
a few dominant terms. We therefore seek a finite angle-distance dictionary
$\{\mathbf{s}(\vartheta_{k},r_{k})\}_{k=1}^{M}$ such that $\hat{\mathbf{h}}=\sum_{k=1}^{M}\varsigma_{k}\mathbf{s}(\vartheta_{k},r_{k})$
approximates $\mathbf{h}$ with a sparse coefficient vector $\bm{\varsigma}$.

The dictionary is designed for Bayesian bandit rather than for constructing
an orthogonal channel basis. The \ac{dft} dictionary can span $\mathbf{h}$,
but near-field and off-grid components can spread over many \ac{dft}
atoms, leading to a high-dimensional and less sparse coefficient posterior.
We therefore cover the hybrid-field steering manifold with angle-distance
atoms so that a steering response $\mathbf{s}(\vartheta,r)$ can be
approximated within a controlled focusing loss. This is similar in
spirit to vector-quantization codebook design, except that the cascaded
channel is represented by a sparse linear combination of a few atoms
rather than by a single codeword. This leads to a design trade-off:
increasing the number of atoms $M$ can reduce the representation
error, but an overly dense dictionary increases mutual coherence and
the cost of \ac{sbl}-based posterior updates, thereby slowing channel
learning and bandit beam adaptation.

\subsubsection{Atom Placement via Uniform Power-Loss Boundaries}

To make this focusing-loss criterion explicit, we introduce auxiliary
atoms and place dictionary atoms using uniform power-loss boundaries.
For an atom $(\vartheta_{m},r_{m})$ and an auxiliary atom $(\vartheta_{j},r_{j})$,
the auxiliary atom marks the boundary of the local representation,
and the parameter $\delta$ is the normalized power response between
these two atoms. After the auxiliary atom is determined, the next
atom is placed on the other side of this boundary such that the same
auxiliary point also satisfies the $\ensuremath{\delta}$ power condition
with respect to the new atom. Conventional designs mainly minimize
mutual coherence \cite{columncovariance,3dB_coherence}, whereas this
criterion directly controls the power loss relevant to pilot and data
beamforming.

In the following, we design an efficient approach to place steering
atoms. According to (\ref{eq:steering=000020vector}) and (\ref{eq:multipath-equiv-chann}),
for atom at $(\vartheta_{m},r_{m})$ and the auxiliary atom $(\vartheta_{j},r_{j})$,
the normalized power response is given by
\begin{align}
G( & \vartheta_{m},r_{m},\vartheta_{j},r_{j})\nonumber \\
 & =\big|\mathbf{s}(\vartheta_{m},r_{m})^{\text{H}}\mathbf{s}(\vartheta_{j},r_{j})\big|^{2}\nonumber \\
 & =\Big|\frac{1}{N}\sum_{n=0}^{N-1}e^{j\pi\big(\phi_{n}(\vartheta_{m},r_{m})-\phi_{n}(\vartheta_{j},r_{j})\big)}\Big|^{2}\nonumber \\
 & =\Big|\frac{1}{N}\sum_{n=0}^{N-1}e^{j\pi\big(n(\vartheta_{j}-\vartheta_{m})+\frac{n^{2}d}{2}(\frac{1-\vartheta_{m}^{2}}{r_{m}}-\frac{1-\vartheta_{j}^{2}}{r_{j}})\big)}\Big|^{2}\label{eq:3dB=000020gain}
\end{align}
where $\phi_{n}(\cdot)$ is given in (\ref{eq:fresnel_approximation}).
The energy focusing criterion sets $G(\vartheta_{m},r_{m},\vartheta_{j},r_{j})=\delta$
between $(\vartheta_{m},r_{m})$ and $(\vartheta_{j},r_{j})$. If
we first enforce $(1-\vartheta_{m}^{2})/r_{m}=(1-\vartheta_{j}^{2})/r_{j}$
in (\ref{eq:3dB=000020gain}), which leads to
\begin{equation}
\frac{1-\vartheta_{m}^{2}}{r_{m}}=\frac{1-\vartheta_{j}^{2}}{r_{j}}\label{eq:same_ring}
\end{equation}
the energy focusing condition from (\ref{eq:3dB=000020gain}) becomes
\begin{align}
 & G(\vartheta_{m},r_{m},\vartheta_{j},r_{j})=\bigg|\frac{1}{N}\sum_{n=0}^{N-1}e^{j\pi n(\vartheta_{j}-\vartheta_{m})}\bigg|^{2}=\delta.\label{eq:gain_angle}
\end{align}
Note that this is the same form as the far-field power response. The
quantity obtained from (\ref{eq:gain_angle}) is the angular offset
from a atom to its auxiliary $\delta$ power boundary, rather than
the direct spacing between two adjacent dictionary atoms. For example,
for $\delta=0.5$, we have $\Delta\vartheta_{\mbox{{\scriptsize 3dB}}}=\vartheta_{m}-\vartheta_{j}\approx0.443\lambda/(Nd)=0.886/N$
for $d=\lambda/2$. 

For the range sampling under a fixed angle $\vartheta_{m}=\vartheta_{j}$,
the energy focusing condition from (\ref{eq:3dB=000020gain}) becomes
\begin{align}
\delta & =\bigg|\frac{1}{N}\sum_{n=0}^{N-1}\text{exp}\left\{ j\frac{\pi n^{2}d}{2}(1-\vartheta_{m}^{2})\Big(\frac{1}{r_{m}}-\frac{1}{r_{j}}\Big)\right\} \bigg|^{2}\nonumber \\
 & \approx|\frac{C(\beta_{\delta})+jS(\beta_{\delta})}{\beta_{\delta}}|^{2}\label{eq:range-equation}
\end{align}
where
\begin{align}
\beta_{\delta}^{2} & \triangleq\frac{N^{2}d}{4}(1-\vartheta_{m}^{2})\left|\frac{1}{r_{m}}-\frac{1}{r_{j}}\right|=\frac{Z}{4}(1-\vartheta_{m}^{2})\left|\frac{1}{r_{m}}-\frac{1}{r_{j}}\right|,\label{eq:range-equation-2}\\
C(\beta_{\delta}) & =\int_{0}^{\beta_{\delta}}\cos\left(\frac{\pi}{2}x^{2}\right)dx,\quad S(\beta_{\delta})=\int_{0}^{\beta_{\delta}}\sin\left(\frac{\pi}{2}x^{2}\right)dx.
\end{align}
Note that the approximation (\ref{eq:range-equation}) is obtained
from \cite{columncovariance}. Here, the equation (\ref{eq:range-equation})
is independent of the angle $\vartheta$ and range $r$. Thus, it
can be easily solved using a numerical search through $\beta_{\delta}$.
For example, for $\delta=0.5$, it is found that $\beta_{3\mathrm{dB}}\approx1.32$.
In addition, the outermost energy boundary relative to the far-field
atom can be found by setting $r=\infty$, which leads to $r_{\max}=r_{j}=\frac{Z}{4\beta_{3\mathrm{dB}}^{2}}(1-\vartheta_{m}^{2})$. 

Specifically, for the dictionary construction, we have the following
two directions. In the sampling rules below, the index $m$ denotes
a sampled angle and $q$ denotes a range center associated with that
angle.

\selectlanguage{american}%
\begin{figure}
\begin{centering}
\includegraphics[width=0.7\columnwidth]{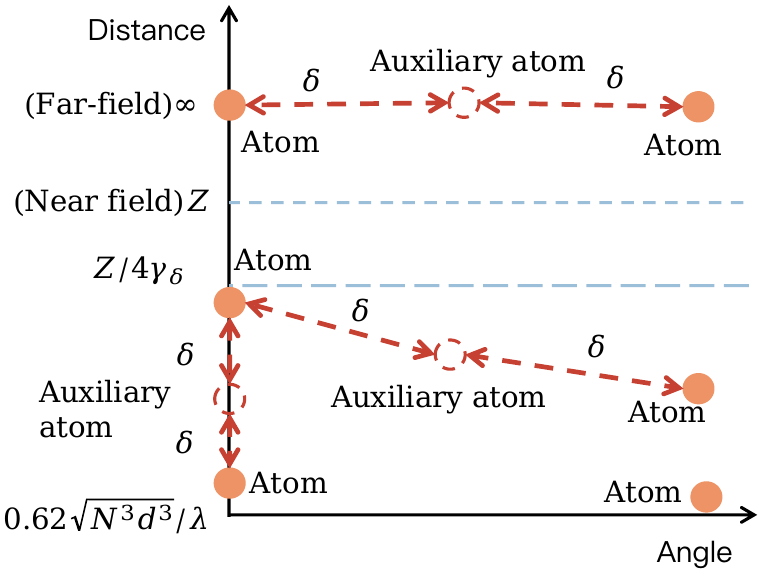}
\par\end{centering}
\caption{\protect\label{fig:codebook}Illustration of the proposed angle-distance
codebook for near-field and far-field, where auxiliary atoms define
uniform power-loss boundaries between adjacent atoms.}
\end{figure}

\selectlanguage{english}%
\begin{itemize}
\item \textbf{Angular-first Sampling:} Starting from $\vartheta=0$, we
sample the angular domain by first identifying the auxiliary $\delta$
power boundary according to (\ref{eq:gain_angle}) and then placing
the next dictionary atom such that the auxiliary atom is shared as
the common boundary of two neighboring atoms. Denote $\mathcal{U}=\{\vartheta_{m}\}$
as the resulting set of dictionary angles. Then, for each angle $\vartheta_{m}\in\mathcal{U}$,
range sampling is performed from the far-field atom $r=\infty$ toward
the prescribed near-field lower range $r_{m,\min}=0.62\sqrt{N^{3}d^{3}/\lambda}$.
The range centers are generated by first solving $\beta_{\delta}$
from (\ref{eq:range-equation}) and then using the inverse-range relation
in (\ref{eq:range-equation-2}). Equivalently, define the $\ensuremath{\mu_{m,q}\triangleq(1-\vartheta_{m}^{2})/r_{m,q}}$
and $\kappa_{\delta}\triangleq4\beta_{\delta}^{2}$. From $(\ref{eq:range-equation-2})$,
an auxiliary atom is located $\kappa_{\delta}/Z$ away from an atom
in the $\mu$ domain. Since two neighboring atoms share the same boundary,
the atom-to-atom spacing is $2\kappa_{\delta}/Z$. Starting from the
far-field atom $\mu_{m,0}=0$, the finite-range centers are placed
at $\mu_{m,q}=2q\kappa_{\delta}/Z$ where $q\ge1$ until the lower
range support is reached, and the set of resulting atoms is $\mathcal{A}^{\text{AF}}=\{\mathbf{s}(\vartheta_{m},r_{m,q})\}_{q\ge0}$
with $q=0$ corresponding to $r=\infty$. 
\item \textbf{Range-first Sampling: }For a fixed $\vartheta_{0}=0$, perform
range sampling from $r_{0,0}=\infty$, towards $r_{0,\min}=0.62\sqrt{N^{3}d^{3}/\lambda}$
according to (\ref{eq:range-equation}). Specifically, for a dictionary
atom $\mathbf{s}(\vartheta_{0},r_{0,q})$, an auxiliary range atom
$\widetilde{r}_{0,q}$ is first determined to satisfy $G(\vartheta_{0},r_{0,q},\vartheta_{0},\widetilde{r}_{0,q})=\delta$.
The next range atom $r_{0,q+1}$ is then placed such that $G(\vartheta_{0},r_{0,q+1},\vartheta_{0},\widetilde{r}_{0,q})=\delta$.
Hence, $\widetilde{r}_{0,q}$ serves as the common $\delta$ power
boundary between two neighboring range atoms. The anchor range set
is first obtained as $\mathcal{R}_{0}=\{r_{0,q}\}_{q\ge0}$, where
$r_{0,0}=\infty$. Then, for each anchor range $r_{0,q}$, joint angle-range
sampling $(\vartheta_{m,q},r_{m,q})$ is performed from the anchor
angle $\vartheta_{0,q}=0$ toward $\vartheta=\pm1$ according to (\ref{eq:same_ring}).
Collecting all generated samples gives the range-first dictionary
$\mathcal{A}^{{\rm RF}}=\{\mathbf{s}(\vartheta_{m,q},r_{m,q})\}_{m,q}.$
\end{itemize}
After collecting the generated angle-range atoms, we relabel them
by a single atom index $k=1,2,\ldots,M$ and form $\mathbf{A}\in\mathbb{C}^{N\times M}$
with columns $\{\mathbf{s}(\vartheta_{k},r_{k})\}_{k=1}^{M}$ generated
by either the angular-first sampling or range-first sampling approach.
Under the adopted sparse dictionary model, we represent the channel
in the dictionary domain as $\mathbf{h}=\mathbf{A}\mathbf{w}$. We
have the following results for the two sampling approaches.
\begin{prop}[\textbf{Sampling-Order Equivalence}]
\label{prop:ring-angule}The two sampling orders generate the same
angle-range dictionaries $\mathcal{A}^{\text{AF}}=\mathcal{A}^{\text{RF}}$
over the same angular grid $\mathcal{U}$, power boundary $\delta$,
and the same system parameters $N$, $d$, and $\lambda$.
\end{prop}
\begin{proof}
See Appendix \ref{sec:proof-ring-ang}.
\end{proof}

\subsection{Implicit Sparse Channel Learning and Phase Configuration}

We adopt \ac{sbl} to estimate the channel $\mathbf{h}=\mathbf{A}\mathbf{w}$
where $\mathbf{A}$ is the proposed energy-focusing dictionary and
$\mathbf{w}\in\mathbb{C}^{M}$ is a sparse vector whose non-zero elements
correspond to the true channel. With the dictionary $\mathbf{A}$,
the received signal in (\ref{eq:received=000020signal}) can be written
as
\begin{align}
\mathbf{y}_{1:t} & =\mathbf{B}_{1:t}^{\top}\mathbf{A}\mathbf{w}+\bm{\xi}_{1:t}=\bm{\Phi}_{1:t}\mathbf{w}+\bm{\xi}_{1:t}
\end{align}
where $\bm{\Phi}_{1:t}\triangleq\mathbf{B}_{1:t}^{\top}\mathbf{A}$
and
\begin{equation}
\bm{\Phi}_{1:t}=[\bm{\phi}_{1}^{\top},\bm{\phi}_{2}^{\top},\dots,\bm{\phi}_{t}^{\top}]^{\top},\quad\bm{\phi}_{t}^{\top}\triangleq\mathbf{b}_{t}^{\top}\mathbf{A},\quad\bm{\phi}_{t}\in\mathbb{C}^{M}.
\end{equation}
Then the scalar observation in slot $t$ is
\begin{equation}
y_{t}=\bm{\phi}_{t}^{\top}\mathbf{w}+\xi_{t}.\label{eq:linear_sparse_model}
\end{equation}
Denoting $\varpi=\sigma^{-2}$, the likelihood of $\mathbf{y}_{1:t}$
is
\begin{equation}
p(\mathbf{y}_{1:t}|\mathbf{w},\varpi)=\mathcal{CN}(\bm{\Phi}_{1:t}\mathbf{w},\varpi^{-1}\mathbf{I}_{t}).\label{eq:bayy-1-1}
\end{equation}
We impose a Gamma distribution on the parameter $\varpi$, given by
\begin{equation}
p(\varpi)=\Gamma(\varpi|1+a,b)=\Gamma(1+a)^{-1}b^{a+1}\varpi^{a}e^{-b\varpi}.
\end{equation}
We also impose a sparsity-promoting prior on the variable $\mathbf{w}$
as
\begin{equation}
p(\mathbf{w}|\bm{\gamma})=\mathcal{CN}(\mathbf{0},\text{diag}(\bm{\gamma}^{-1}))
\end{equation}
where $\bm{\gamma}=[\gamma_{1},\gamma_{2},...,\gamma_{M}]^{\top}$,
whose entries have Gamma hyperpriors as
\begin{equation}
p(\bm{\gamma};c,d)=\prod_{m=1}^{M}\Gamma(\gamma_{m}|1+c,d).\label{eq:gamma-1-1-1}
\end{equation}

For fixed hyperparameters $(\text{diag}(\bm{\gamma}),\varpi)$, the
posterior remains Gaussian $p(\mathbf{w}\mid\mathbf{y}_{1:t},\bm{\gamma},\varpi)=\mathcal{CN}(\bm{\mu}_{\mathbf{w},t},\mathbf{\Sigma}_{\mathbf{w},t})$
with
\begin{equation}
\mathbf{\Sigma}_{\mathbf{w},t}=\left(\varpi\bm{\Phi}_{1:t}^{\mathrm{H}}\bm{\Phi}_{1:t}+\text{diag}(\bm{\gamma})\right)^{-1},\label{eq:sparse_posterior_cov_batch}
\end{equation}
\begin{equation}
\bm{\mu}_{\mathbf{w},t}=\varpi\mathbf{\Sigma}_{\mathbf{w},t}\bm{\Phi}_{1:t}^{\mathrm{H}}\mathbf{y}_{1:t}.\label{eq:sparse_posterior_mean_batch}
\end{equation}
Conditioned on the current hyperparameters, the induced posterior
of the cascaded channel is
\begin{equation}
\mathbf{h}\mid\mathcal{F}_{t}\sim\mathcal{CN}\bigl(\mathbf{A}\bm{\mu}_{\mathbf{w},t},\mathbf{A}\mathbf{\Sigma}_{\mathbf{w},t}\mathbf{A}^{\mathrm{H}}\bigr).\label{eq:sparse_channel_posterior}
\end{equation}

Here, we adopt a warm-start SBL refinement, where the \ac{pdf} at
slot $t$ is initialized by the converged hyperparameters from slot
$t-1$. Moreover, with these fixed initial hyperparameters, the posterior
covariance and mean are first initialized through a Woodbury-based
rank-one update from the previous slot posterior. The Thompson-sampling
step is then implemented in the coefficient domain. At the beginning
of slot $t$, we draw
\begin{equation}
\mathbf{w}'_{t}\sim\mathcal{CN}(\bm{\mu}_{\mathbf{w},t-1},\mathbf{\Sigma}_{\mathbf{w},t-1}),\hspace*{1em}\mathbf{h}'_{t}=\mathbf{A}\mathbf{w}_{t}'
\end{equation}
and choose the RIS reflection vector according to the same phase-alignment
rule as in Section \ref{sec:Bayesian-Bandit-Beamforming} as
\begin{equation}
\mathbf{b}_{t}\in\arg\max_{\mathbf{b}\in\mathcal{B}}\left|\mathbf{b}^{\top}\mathbf{h}_{t}'\right|^{2}
\end{equation}
where element-wise solution is
\begin{equation}
[\mathbf{b}_{t}]_{n}=\exp\big(-j\angle([\mathbf{h}_{t}']_{n})\big),\qquad n=1,\ldots,N.
\end{equation}

\subsection{Hyper-parameter Refinement and Computational Complexity}

The posterior above is conditioned on the current hyperparameters
$(\text{diag}(\bm{\gamma}),\varpi)$. However, these hyperparameters
are not kept fixed across slots. At slot $t$, we initialize the hyperparameters
from slot $t-1$ as $\varpi_{t}^{(0)}=\varpi_{t-1}$ and $\gamma_{m,t}^{(0)}=\gamma_{m,t-1}$.
After obtaining the received signal $y_{t}$ with sensing row $\bm{\phi}_{t}^{\top}=\mathbf{b}_{t}^{\top}\mathbf{A}$,
the initial posterior at slot $t$ can be obtained from the previous-slot
posterior through a Woodbury-based rank-one update as
\begin{align}
\mathbf{\Sigma}_{\mathbf{w},t}^{(0)}= & \mathbf{\Sigma}_{\mathbf{w},t-1}-\mathbf{\Sigma}_{\mathbf{w},t-1}\bm{\phi}_{t}^{*}\left(\varpi_{t-1}^{-1}+\bm{\phi}_{t}^{\top}\mathbf{\Sigma}_{\mathbf{w},t-1}\bm{\phi}_{t}^{*}\right)^{-1}\nonumber \\
 & \:\times\bm{\phi}_{t}^{\top}\mathbf{\Sigma}_{\mathbf{w},t-1},\label{eq:warm_woodbury_sigma}\\
\bm{\mu}_{\mathbf{w},t}^{(0)}= & \bm{\mu}_{\mathbf{w},t-1}+\mathbf{\Sigma}_{\mathbf{w},t-1}\bm{\phi}_{t}^{*}\left(\varpi_{t-1}^{-1}+\bm{\phi}_{t}^{\top}\mathbf{\Sigma}_{\mathbf{w},t-1}\bm{\phi}_{t}^{*}\right)^{-1}\nonumber \\
 & \:\times\left(y_{t}-\bm{\phi}_{t}^{\top}\bm{\mu}_{\mathbf{w},t-1}\right).\label{eq:warm_woodbury_mu}
\end{align}

The pair $(\bm{\mu}_{\mathbf{w},t}^{(0)},\mathbf{\Sigma}_{\mathbf{w},t}^{(0)})$
is then used as an initialization, after which the inner warm-start
SBL iterations refine the posterior and hyperparameters until convergence.
Starting from (\ref{eq:warm_woodbury_sigma})$-$(\ref{eq:warm_woodbury_mu}),
we run inner SBL iteration indexed by $k=1,2,\dots$ until convergence.
At the inner iteration $k$, the posterior covariance and mean are
updated as
\begin{align}
\mathbf{\Sigma}_{\mathbf{w},t}^{(k)} & =\left(\varpi_{t}^{(k-1)}\bm{\Phi}_{1:t}^{\mathrm{H}}\bm{\Phi}_{1:t}+\mathrm{diag}\bigl(\bm{\gamma}_{t}^{(k-1)}\bigr)\right)^{-1},\label{eq:warm_sbl_sigma}\\
\bm{\mu}_{\mathbf{w},t}^{(k)} & =\varpi_{t}^{(k-1)}\mathbf{\Sigma}_{\mathbf{w},t}^{(k)}\bm{\Phi}_{1:t}^{\mathrm{H}}\mathbf{y}_{1:t}.\label{eq:warm_sbl_mu}
\end{align}
Then the hyperparameters are updated as
\begin{align}
\varpi_{t}^{(k)} & \leftarrow\frac{t+a}{b+\left\Vert \mathbf{y}_{1:t}-\bm{\Phi}_{1:t}\bm{\mu}_{\mathbf{w},t}^{(k)}\right\Vert _{2}^{2}+f_{\varpi}^{(k)}},\label{eq:warm_sbl_alpha}\\
\gamma_{m,t}^{(k)} & \leftarrow\frac{c+1}{d+\left|[\bm{\mu}_{\mathbf{w},t}^{(k)}]_{m}\right|^{2}+[\mathbf{\Sigma}_{\mathbf{w},t}^{(k)}]_{m,m}}\label{eq:warm_sbl_delta}
\end{align}
where
\[
f_{\varpi}^{(k)}=(\varpi_{t}^{(k-1)})^{-1}\sum_{m=1}^{M}\left(1-\gamma_{m,t}^{(k-1)}[\mathbf{\Sigma}_{\mathbf{w},t}^{(k)}]_{m,m}\right).
\]

The inner iterations in $(\ref{eq:warm_sbl_sigma})-(\ref{eq:warm_sbl_delta})$
are repeated until convergence. Let $(\varpi_{t},\{\gamma_{m,t}\}_{m=1}^{M},\bm{\mu}_{\mathbf{w},t},\mathbf{\Sigma}_{\mathbf{w},t})$
denote the converged result at slot $t$. This posterior is then used
to form the Thompson-sampling distribution for the next slot, while
$(\varpi_{t},\{\gamma_{m,t}\}_{m=1}^{M})$ are passed to slot $t+1$
as the warm-start initialization.

Let $K_{t}$ denote the number of inner SBL iterations at slot $t$.
With the Woodbury-based warm-start initialization in $(\ref{eq:warm_woodbury_sigma})\text{--}(\ref{eq:warm_woodbury_mu})$,
the initial posterior update requires $\mathcal{O}(M^{2})$ operations,
whereas each subsequent inner refinement step involves updating the
posterior covariance and mean which requires $\mathcal{O}(M^{3})$
operations. Therefore, the overall per-slot complexity is given by
$\mathcal{O}\big(M^{2}+K_{t}M^{3}+NM\big)$, where the final $\mathcal{O}(NM)$
term accounts for the reconstruction of the channel in the original
domain.

\begin{algorithm}[t]
\caption{SBL-based Thompson sampling with warm starts}
\label{alg:sbl-ts}
\begin{algorithmic}[1]
\REQUIRE Dictionary $\mathbf A$, weakly informative hyperparameters, noise initialization $\varpi_0$, initial precisions $\boldsymbol\gamma_0$, maximum inner iterations $K_{\max}$, stopping tolerance $\epsilon_{\rm tol}$, and horizon $T$.
\FOR{$t=1,2,\ldots,T$}
\STATE Draw $\mathbf w'_t$ from the current coefficient posterior and set $\mathbf h'_t=\mathbf A\mathbf w'_t$.
\STATE Choose $[\mathbf b_t]_n=\exp\{-j\angle([\mathbf h'_t]_n)\}$.
\STATE Use $\mathbf b_t$ for pilot and payload symbols and observe $y_t$.
\STATE Apply the Woodbury warm-start update for the coefficient posterior.
\STATE Refine $\varpi$, $\boldsymbol\gamma$, posterior mean, and posterior covariance by SBL iterations.
\STATE Stop when the relative change is below $\epsilon_{\rm tol}$ or $K_{\max}$ is reached.
\ENDFOR
\ENSURE RIS configurations and sparse posterior sequence.
\end{algorithmic}
\end{algorithm}

\section{Simulation Results}

This section evaluates the proposed MMSE-based and SBL-based Thompson-sampling
policies for RIS-assisted wireless communication. The experiments
compare convergence under different channel regimes, sensitivity to
\ac{snr}, multi-armed bandit (MAB)-based baselines, and dictionary
performance.

\subsection{Environment Setup}

\begin{figure*}[t]
\centering
\subfigure[]{\includegraphics[width=0.32\textwidth]{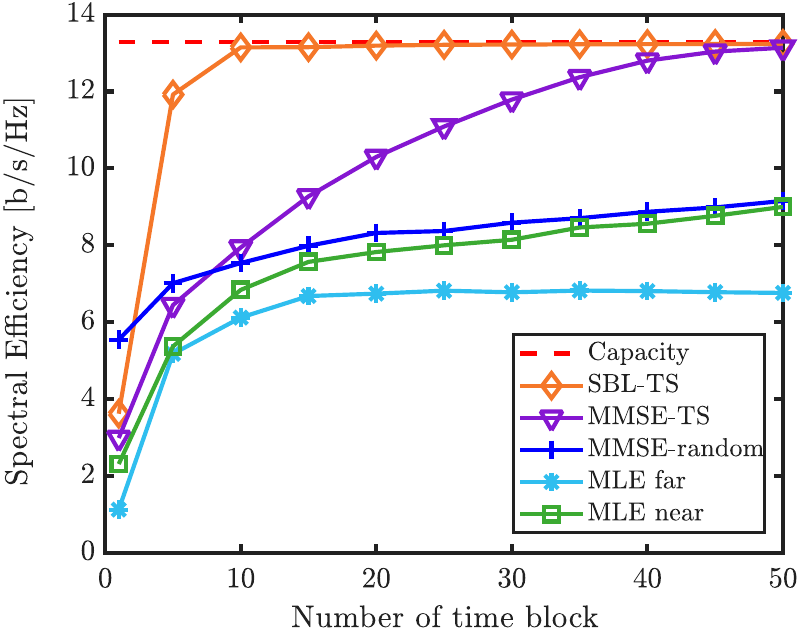}\label{fig:coveragence-1}}\hfill
\subfigure[]{\includegraphics[width=0.32\textwidth]{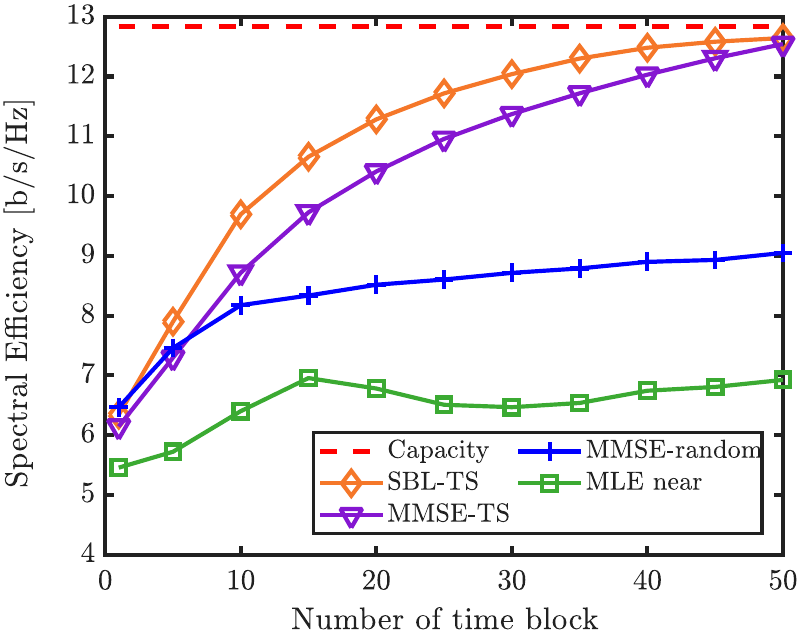}\label{fig:convergence-2}}\hfill
\subfigure[]{\includegraphics[width=0.32\textwidth]{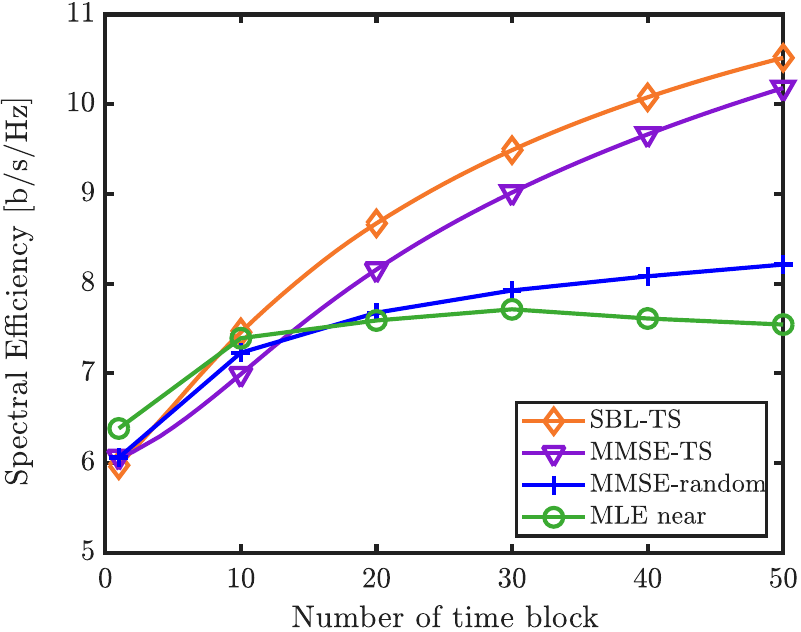}\label{fig:convergence-3}}
\caption{Convergence performance under three channel settings: (a) instantaneous spectral efficiency versus time block for the \ac{los} scenario; (b) instantaneous spectral efficiency versus time block for the multipath scenario; (c) average spectral efficiency versus time block for the Rayleigh fading scenario.}
\label{fig:convergence-three-panel}
\end{figure*}

Unless otherwise stated, we consider a narrowband RIS-assisted downlink
at 28 GHz with half-wavelength element spacing $d=\lambda/2$ and
$N=100$ reflecting elements. The simulations cover a \ac{los}-dominant
sparse cascaded channel, a hybrid multipath channel with near- and
far-field components, and an unstructured Rayleigh fading channel.

For the geometric regimes, the path angles in the setup are drawn
from $\theta\sim\mathcal{U}[-\pi/2,\pi/2]$. The path distances are
drawn from $r\sim\mathcal{U}[0.62\sqrt{N^{3}\lambda^{2}/8},10Z]$,
where $Z=2N^{2}d^{2}/\lambda=N^{2}\lambda/2$ is the Rayleigh distance
\cite{Guerra2021distance}. For the sparse geometric channels, we
independently generate the BS-RIS and RIS-user channel vectors as
sums of near-/far-field steering responses, and then form the cascaded
channel by element-wise multiplication. The geometric parameters are
randomly drawn within the prescribed angular and distance ranges.
The Rayleigh fading channel is generated as $\mathbf{h}\sim\mathcal{CN}(\mathbf{0},\mathbf{I}_{N})$.

Unless otherwise specified, the pilot and data SNRs are identical
and set to 0 dB. Spectral efficiency is evaluated using the effective-rate
definition in (\ref{eq:effective_rate}). We set $\eta=1$ and report
the normalized spectral efficiency; using any other common payload
fraction would only rescale all curves by the same factor. Unless
a figure caption or table states otherwise, the curves are averaged
over 1000 independent Monte Carlo realizations.

\subsection{Baseline Schemes}

We compare with the following three baselines:
\begin{itemize}
\item Parametric \ac{mle} \cite{MLE2}: This baseline uses a model-based
channel estimator. The far-field variant assumes planar wavefronts
and estimates angular parameters only, while the near-field variant
uses spherical wavefronts and estimates both angular and range parameters.
\item Non-adaptive MMSE: This baseline uses randomly generated RIS configurations
for channel probing and then designs the RIS configuration based on
the resulting \ac{mmse} channel estimate. In the figures, this baseline
is denoted as MMSE-random.
\item $\varepsilon$-greedy MAB \cite{Mohamed2021B}: This finite-codebook
bandit baseline selects a random RIS configuration with probability
$\varepsilon=0.3$, and otherwise exploits the configuration with
the best observed performance so far. 
\end{itemize}

\subsection{Evaluation Metrics}

We use two metrics. The first metric is the block-averaged spectral
efficiency, referred to simply as spectral efficiency hereafter. To
account for the estimate-then-optimize baselines, the spectral efficiency
is evaluated per time block by averaging over two RIS configurations:
$R_{t,\text{b}}=\frac{1}{2}\sum_{\ell=1}^{2}U_{t}(\mathbf{b}_{t}^{(\ell)},\mathbf{h})$
where $U_{t}(\cdotp)$ is defined in (\ref{eq:effective_rate}). Here,
$\mathbf{b}_{t}^{(1)}$ and $\mathbf{b}_{t}^{(2)}$ denote the two
RIS configurations used within time block $t$. For estimate-then-optimize
baselines, they correspond to the pilot and data-transmission configurations,
respectively; for the proposed policy, they correspond to two pilot
and data-sharing transmissions.

The second metric is the channel-learning error. When a curve or table
is labeled \ac{nmse}, it is defined as $\mathrm{NMSE}=||\hat{\mathbf{h}}_{t}-\mathbf{h}||_{2}^{2}/||\mathbf{h}||_{2}^{2}.$ 

\subsection{Convergence Performance and Comparison with Pilot-Based Methods}

Figs. \ref{fig:coveragence-1}-\ref{fig:convergence-3} illustrate
the evolution of spectral efficiency under \ac{los}, hybrid multipath,
and Rayleigh fading conditions, respectively. The Capacity curve serves
as an upper bound assuming perfect-\ac{csi} and optimal beam alignment. 

As shown in Fig. \ref{fig:coveragence-1}, in the \ac{los}-dominant
regime, the proposed \ac{sbl}-based Thompson sampling approaches
the capacity reference within 10 time blocks. At 50 time blocks, it
is about 44\% higher than MMSE-random and MLE-near, and more than
95\% higher than MLE-far. The rapid convergence can be attributed
to the strong spatial sparsity of the channel, which allows \ac{sbl}
inference to rapidly concentrate the posterior on a low-dimensional
support.

As shown in Fig. \ref{fig:convergence-2}, in the hybrid multipath
scenario the convergence rate naturally degrades due to the increased
channel rank. Nevertheless, both the \ac{mmse}-based and \ac{sbl}-based
Thompson-sampling schemes consistently outperform non-adaptive \ac{mmse}
and \ac{mle}-based baselines. 

Figs. \ref{fig:coveragence-1} and \ref{fig:convergence-2} report
instantaneous spectral efficiency to show adaptation speed, whereas
 Fig. \ref{fig:convergence-3} uses the averaged spectral efficiency
$1/t\sum_{\tau=1}^{t}R_{\tau,\text{b}}$ for the more fluctuating
Rayleigh fading case.

As shown in Fig. \ref{fig:convergence-3}, in the Rayleigh fading
case, the geometric sparsity model is no longer well matched, so the
gap between the proposed schemes is smaller. Nevertheless, MMSE-based
Thompson sampling and SBL-based Thompson sampling still improve over
random-MMSE by about 24\% and 28\%, and over MLE-near by about 30\%
and 39\%, respectively, indicating that posterior sampling remains
useful even when the geometric sparse model is mismatched.

\subsection{Comparison under Different SNRs}

\selectlanguage{american}%
\begin{figure}
\begin{centering}
\includegraphics[width=0.88\columnwidth]{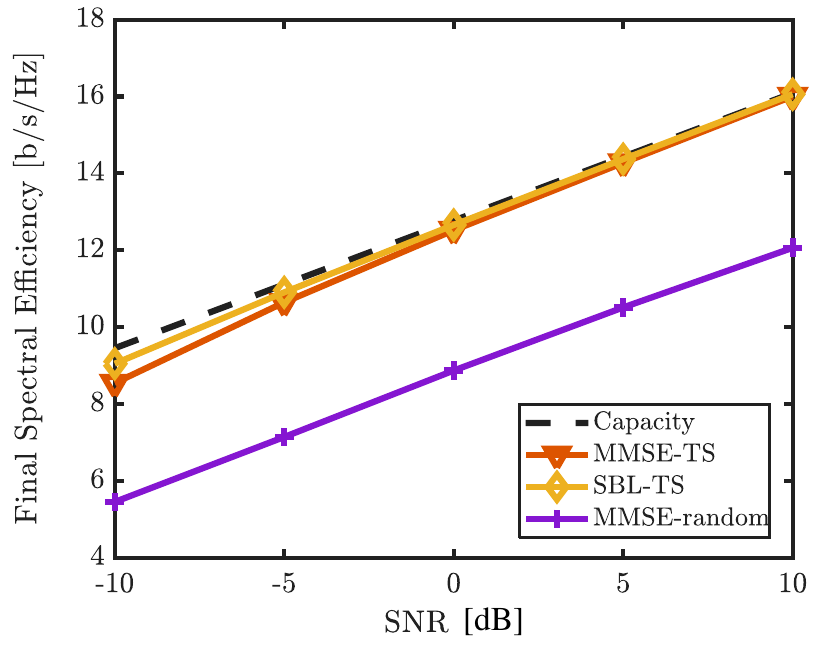}
\par\end{centering}
\caption{\protect\label{fig:snr_sweep_final_se}Final spectral efficiency versus
SNR in the hybrid multipath scenario with $N=100$, $L=3$, $T=100$.
The SNR is varied from -10 dB to 10 dB.}
\end{figure}

Fig. \ref{fig:snr_sweep_final_se} evaluates the sensitivity to the
pilot and data SNR. Even at -10 dB, SBL-based Thompson sampling achieves
9.0 b/s/Hz, only 4.3\% below the capacity reference, while improving
MMSE-based Thompson sampling and MMSE-random by about 5.9\% and 66\%,
respectively. At 10 dB, it nearly matches the capacity reference,
whereas MMSE-random remains at 12.1 b/s/Hz. These results indicate
that higher SNR improves posterior concentration for the adaptive
schemes, whereas non-adaptive algorithm remains limited by its lack
of channel-dependent RIS pattern selection. The results demonstrate
the finite-horizon robustness of the proposed policies under the multipath
setting.
\selectlanguage{english}%

\subsection{Comparison with a MAB Baseline}

\selectlanguage{american}%
\begin{figure}
\begin{centering}
\includegraphics[width=0.88\columnwidth]{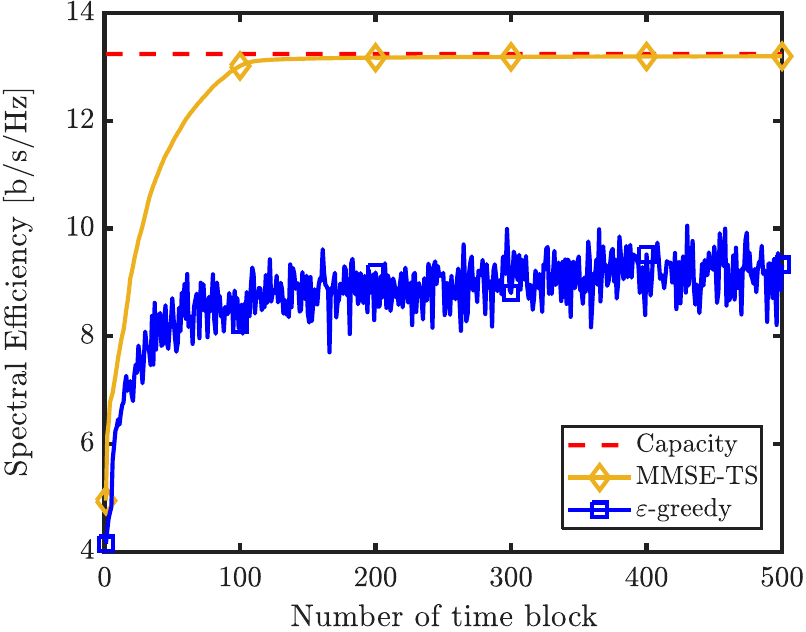}
\par\end{centering}
\caption{\protect\label{fig:MAB=000020compare}Average spectral efficiency
versus time block for the MAB baseline.}
\end{figure}

\selectlanguage{english}%
Fig. \ref{fig:MAB=000020compare} compares the proposed Bayesian Thompson-sampling
strategies with the classical $\varepsilon$-greedy MAB baseline.
Although this baseline gradually improves performance through exploration
and exploitation, the convergence is substantially slower.

At $t=500$, the proposed curve is about $13.2$ b/s/Hz, whereas $\varepsilon$-greedy
remains below $9.5$ b/s/Hz. The resulting gain is approximately 40\%.
This gap arises because the MAB baseline treats each RIS configuration
as an independent arm, while the proposed \ac{mmse}-based posterior
update exploits the channel correlation among beam patterns. 

\subsection{Comparison with Dictionary-Based Baselines}

\begin{table}[ht]
\centering
\caption{NMSE comparison of different dictionaries}
\label{tab:converge}
\begin{tabular}{ccccc}
    \toprule
    \textbf{Scenario} & \textbf{Proposed} & \textbf{Polar} & \textbf{3 dB coherence} & \textbf{DFT} \\
    \midrule
    $L=5$ & $-7.3\,\mathrm{dB}$ & $-6.9\,\mathrm{dB}$ & $-1.4\,\mathrm{dB}$ & $-5.1\,\mathrm{dB}$ \\
    $L=10$ & $-7.0\,\mathrm{dB}$ & $-5.9\,\mathrm{dB}$ & $-2.1\,\mathrm{dB}$ & $-4.7\,\mathrm{dB}$ \\
    $L=15$ & $-7.6\,\mathrm{dB}$ & $-6.3\,\mathrm{dB}$ & $-2.7\,\mathrm{dB}$ & $-4.9\,\mathrm{dB}$ \\
    Rayleigh & $-3.1\,\mathrm{dB}$ & $-3.3\,\mathrm{dB}$ & $-2.3\,\mathrm{dB}$ & $-2.4\,\mathrm{dB}$ \\
    \bottomrule
\end{tabular}
\end{table}

\selectlanguage{american}%
This subsection evaluates both the sparse representation accuracy
of the proposed energy-focusing dictionary with different baseline
dictionaries and their resulting online estimation performance. We
use Orthogonal Matching Pursuit (OMP) with a fixed sparsity level.
For a channel with $L$ effective cascaded components, OMP selects
$L$ atoms. We compare it with a far-field DFT dictionary \cite{columncovariance},
a polar-domain near-field dictionary \cite{columncovariance}, and
the 3 dB coherence dictionary \cite{3dB_coherence}.

\selectlanguage{english}%
As shown in Table \ref{tab:converge}, the proposed dictionary consistently
reduces the \ac{nmse} in the sparse geometric channels. Compared
with the \ac{dft} dictionary, it provides gains of about 2.2, 2.3,
and 2.7 dB for $L=5$, $L=10$, and $L=15$, respectively, and its
gain over the 3 dB coherence dictionary is around 5 dB. These improvements
come from the energy-focusing atom placement, which better captures
the hybrid near-/far-field structure of the cascaded channel. In the
Rayleigh fading case, the channel does not follow a sparse geometric
model, and the NMSE values of different dictionaries become close.

Fig. \ref{fig:compare-dictionary-1} shows the corresponding online
estimation behavior. At 100 time blocks, the proposed dictionary attains
an \ac{nmse} of $-11.5$ dB, outperforming the polar-domain, \ac{dft},
and 3 dB coherence dictionaries by about 1 dB, 4 dB, and 8 dB, respectively.
It uses 333 atoms, compared with 100 atoms for the \ac{dft} dictionary,
519 atoms for the polar-domain dictionary, and 115 atoms for the 3
dB coherence dictionary. Thus, compared with the polar-domain dictionary,
the proposed construction reduces the dictionary size by about 36\%
while still achieving better estimation accuracy. This provides a
more compact posterior representation for online bandit learning.

\selectlanguage{american}%
\begin{figure}
\begin{centering}
\includegraphics[width=0.88\columnwidth]{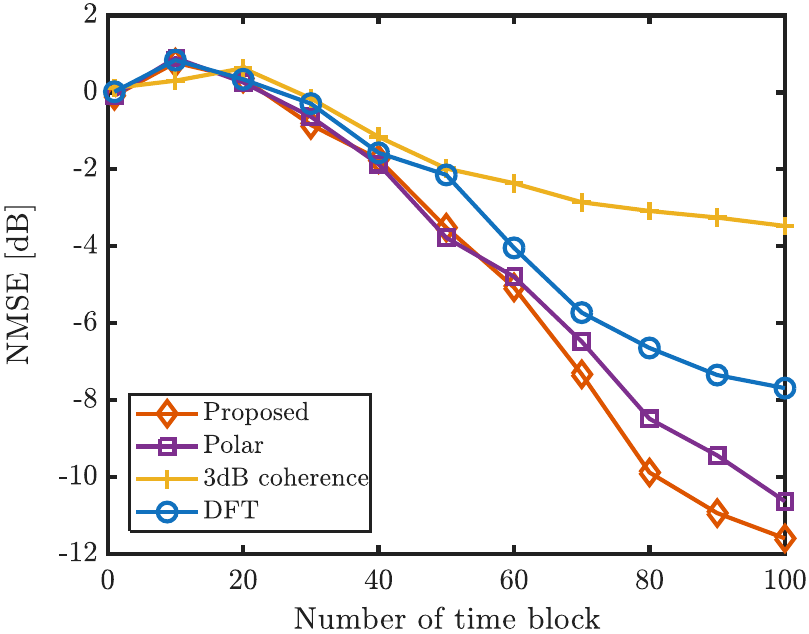}
\par\end{centering}
\caption{\protect\label{fig:compare-dictionary-1}NMSE versus the number of
time blocks.}
\end{figure}

\selectlanguage{english}%

\subsection{Posterior Scaling Factor}

\selectlanguage{american}%
\begin{figure}
\begin{centering}
\includegraphics[width=0.88\columnwidth]{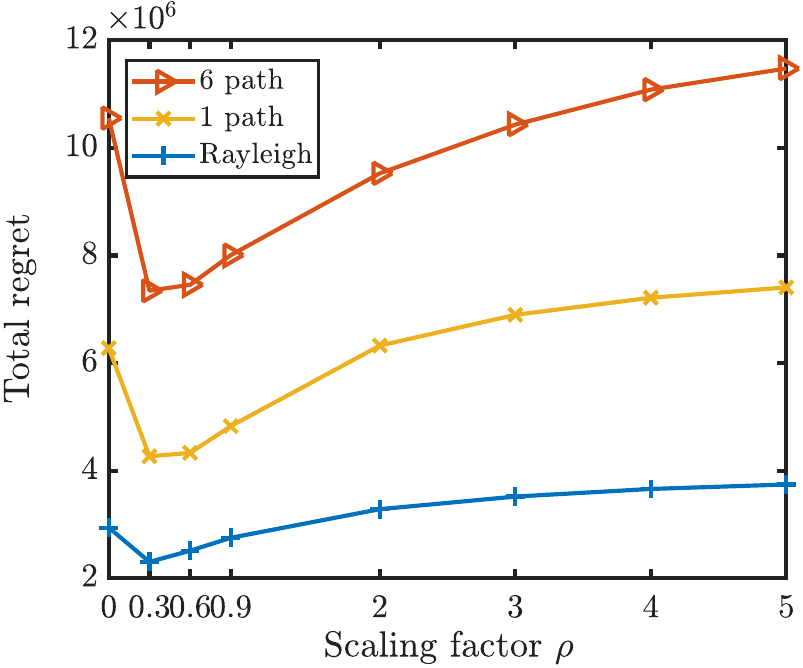}
\par\end{centering}
\caption{\protect\label{fig:scaling=000020factor}Cumulative received-power
regret versus posterior scaling factor.}
\end{figure}

Since Theorem $\ref{thm:Bayesian-regret-decomposition}$ shows that
the regret is governed by posterior uncertainty contraction, we study
the posterior scaling factor $\rho$ to evaluate how the exploration
level induced by the sampling covariance affects the finite-horizon
regret. Specifically, Fig. $\ref{fig:scaling=000020factor}$ reports
the cumulative received-power regret when the sampled channel is drawn
as $\mathbf{h}_{t}'\sim\mathcal{CN}(\hat{\mathbf{h}}_{t-1},\rho\bm{\Sigma}_{t-1})$
by MMSE-based Thompson sampling, where $\rho$ controls the exploration
level. The experiment is conducted with $N=100$, SNR = 0 dB, and
$T=200$ under the one-path, six-path, and Rayleigh fading settings.

The resulting curve exhibits a U-shaped trend with respect to the
scaling factor $\rho$. When $\rho$ is too small, the sampled channel
is close to the posterior mean, and the policy becomes insufficiently
exploratory. When $\rho$ is too large, the sampled channel is overly
exploratory, and the selected RIS patterns may deviate substantially
from the current channel estimate. In the figure, the lowest cumulative
regret is attained around $\rho=0.3$, which reduces the regret by
17\% compared with the greedy posterior-mean case $\rho=0$. In contrast,
over-scaling the posterior uncertainty to $\rho=5$ increases the
regret by 62\%. These results show that a moderate scaling factor
improves the exploration-exploitation balance and reduces regret,
while too little or too much scaling leads to higher regret.
\selectlanguage{english}%

\section{Conclusion}

This paper studied RIS phase-shift configuration under short coherence
times and hybrid near-/far-field propagation. We formulated RIS control
as a Bayesian bandit beamforming problem with implicit channel learning,
where each selected RIS pattern both probes the cascaded channel and
supports payload transmission. Based on this formulation, we developed
an MMSE-based Thompson-sampling policy and a hybrid-field SBL extension
that exploit sparse angle-distance structure for high frequency communications.
The analysis links Bayesian received-power regret to posterior covariance
contraction and establishes a conditional sublinear Bayesian-regret
guarantee under a cumulative information-growth condition. Simulation
results demonstrate higher performance than the considered baselines,
rapid convergence in spectral efficiency, and improved robustness
in hybrid multipath and Rayleigh fading scenarios.

\bibliographystyle{IEEEtran}
\bibliography{bib/JCgroup_reinforceRIS}

\clearpage

\appendices

\section{Proof of Lemma \ref{lem:One-slot-Mismatching-Bound}~\protect\label{sec:proof_one_slot_mismatching}}

By definition of the phase-alignment vector, $[\mathbf{b}_{*}]_{n}=e^{-j\angle(h_{n})}$
for each $n$. Therefore,
\begin{equation}
\mathbf{b}_{*}^{\top}\mathbf{h}=\sum_{n=1}^{N}e^{-j\angle(h_{n})}h_{n}=\sum_{n=1}^{N}|h_{n}|=\|\mathbf{h}\|_{1}
\end{equation}
which is real and nonnegative, and hence $\big|\mathbf{b}_{*}^{\top}\mathbf{h}\big|=\|\mathbf{h}\|_{1}$.
In this way we have, 
\begin{equation}
|\mathbf{b}_{*}^{\top}\mathbf{h}|^{2}-|\mathbf{b}_{t}^{\top}\mathbf{h}|^{2}=(\|\mathbf{h}\|_{1}-|\mathbf{b}_{t}^{\top}\mathbf{h}|)(\|\mathbf{h}\|_{1}+|\mathbf{b}_{t}^{\top}\mathbf{h}|).\label{eq:app_factorization}
\end{equation}

Since $|\mathbf{b}_{t}^{\top}\mathbf{h}|\le\|\mathbf{h}\|_{1}$, we
have
\begin{equation}
\|\mathbf{h}\|_{1}+|\mathbf{b}_{t}^{\top}\mathbf{h}|\le2\|\mathbf{h}\|_{1}.\label{eq:app_ac_upper}
\end{equation}
It remains to bound the first factor, since we have $\mathbf{b}_{t}^{\top}\mathbf{h}'_{t}=\|\mathbf{h}'_{t}\|_{1}$
which is real and nonnegative. Therefore,
\begin{align}
\|\mathbf{h}\|_{1}-|\mathbf{b}_{t}^{\top}\mathbf{h}| & \le\|\mathbf{h}\|_{1}-\Re\left\{ \mathbf{b}_{t}^{\top}\mathbf{h}\right\} \nonumber \\
 & =\|\mathbf{h}\|_{1}-\Re\left\{ \mathbf{b}_{t}^{\top}\mathbf{h}+\mathbf{b}_{t}^{\top}\mathbf{h}'_{t}-\mathbf{b}_{t}^{\top}\mathbf{h}'_{t}\right\} \nonumber \\
 & =\|\mathbf{h}\|_{1}-\|\mathbf{h}'_{t}\|_{1}+\Re\left\{ \mathbf{b}_{t}^{\top}(\mathbf{h}'_{t}-\mathbf{h})\right\} \nonumber \\
 & \le\big|\|\mathbf{h}\|_{1}-\|\mathbf{h}'_{t}\|_{1}\big|+\left|\Re\left\{ \mathbf{b}_{t}^{\top}(\mathbf{h}'_{t}-\mathbf{h})\right\} \right|\nonumber \\
 & \le\|\mathbf{h}-\mathbf{h}'_{t}\|_{1}+\|\mathbf{h}-\mathbf{h}'_{t}\|_{1}=2\|\mathbf{h}-\mathbf{h}'_{t}\|_{1}\label{eq:app_a_minus_c}
\end{align}
where the last inequality uses reverse triangle inequality \cite{triangle_inequality}
and the fact that each entry of $\mathbf{b}_{t}$ has unit modulus
and utilizing Hölder inequality \cite{real_analysis}. Substituting
$(\ref{eq:app_ac_upper})$ and $(\ref{eq:app_a_minus_c})$ into $(\ref{eq:app_factorization})$,
we obtain
\begin{equation}
|\mathbf{b}_{*}^{\top}\mathbf{h}|^{2}-|\mathbf{b}_{t}^{\top}\mathbf{h}|^{2}\le4\|\mathbf{h}\|_{1}\|\mathbf{h}-\mathbf{h}'_{t}\|_{1}.
\end{equation}

Finally, using $\|\mathbf{x}\|_{1}\le\sqrt{N}\|\mathbf{x}\|_{2}$,
we further get
\begin{align}
|\mathbf{b}_{*}^{\top}\mathbf{h}|^{2}-|\mathbf{b}_{t}^{\top}\mathbf{h}|^{2} & \le4(\sqrt{N}\|\mathbf{h}\|_{2})(\sqrt{N}\|\mathbf{h}-\mathbf{h}'_{t}\|_{2}),\nonumber \\
 & \le4N\|\mathbf{h}\|_{2}\|\mathbf{h}-\mathbf{h}'_{t}\|_{2},
\end{align}
which proves $(\ref{eq:one_slot_decomposition})$.

\section{Proof of Theorem \ref{thm:Bayesian-regret-decomposition} ~\protect\label{sec:proof_bayesian_regret_decomposition}}

Recall the one-slot instantaneous received-power regret as $\Delta_{t}\triangleq|\mathbf{b}_{*}^{\top}\mathbf{h}|^{2}-|\mathbf{b}_{t}^{\top}\mathbf{h}|^{2}.$
Then $\mathcal{R}_{T}=\sum_{t=1}^{T}\mathbb{E}[\Delta_{t}]$. By Lemma
$\ref{lem:One-slot-Mismatching-Bound}$, we obtain $\Delta_{t}\le4N\|\mathbf{h}\|_{2}\|\mathbf{h}-\mathbf{h}'_{t}\|_{2}.$
Conditioning on $\mathcal{F}_{t-1}$, taking the expectation and applying
Cauchy-Schwarz gives
\begin{align}
\mathbb{E}\left[\Delta_{t}\mid\mathcal{F}_{t-1}\right] & \le4N\Big(\mathbb{E}[\|\mathbf{h}\|_{2}^{2}\mid\mathcal{F}_{t-1}]\Big)^{1/2}\nonumber \\
 & \quad\times\Big(\mathbb{E}[\|\mathbf{h}-\mathbf{h}'_{t}\|_{2}^{2}\mid\mathcal{F}_{t-1}]\Big)^{1/2}.\label{eq:app_conditional_cs}
\end{align}

Under the Gaussian posterior update, $\mathbf{h}\mid\mathcal{F}_{t-1}\sim\mathcal{CN}(\bm{\mu}_{t-1},\bm{\Sigma}_{t-1}),$
and sample $\mathbf{h}'_{t}$ is drawn independently from the same
conditional distribution. Therefore, $\mathbf{h}$ and $\mathbf{h}'_{t}$
are conditionally i.i.d. Given $\mathcal{F}_{t-1}$, we have
\begin{equation}
\mathbb{E}\left[\|\mathbf{h}-\mathbf{h}'_{t}\|_{2}^{2}\mid\mathcal{F}_{t-1}\right]=2\text{tr}(\bm{\Sigma}_{t-1}).\label{eq:app_same_posterior}
\end{equation}

Since $\bm{\Sigma}_{t-1}$ is determined by $\mathcal{F}_{t-1}$,
the trace term in $(\ref{eq:app_same_posterior})$ is $\mathcal{F}_{t-1}$-measurable.

Substituting $(\ref{eq:app_same_posterior})$ into $(\ref{eq:app_conditional_cs})$
yields
\begin{equation}
\mathbb{E}\left[\Delta_{t}\mid\mathcal{F}_{t-1}\right]\le4\sqrt{2}N\Big(\mathbb{E}[\|\mathbf{h}\|_{2}^{2}\mid\mathcal{F}_{t-1}]\Big)^{1/2}\bigl(\text{tr}(\bm{\Sigma}_{t-1})\bigr)^{1/2}.
\end{equation}

Taking expectation again and applying Cauchy-Schwarz together with
the tower property gives
\begin{align}
\mathbb{E}[\Delta_{t}] & \le4\sqrt{2}N\mathbb{E}\!\left[\Big(\mathbb{E}[\|\mathbf{h}\|_{2}^{2}\mid\mathcal{F}_{t-1}]\Big)^{1/2}\bigl(\text{tr}(\bm{\Sigma}_{t-1})\bigr)^{1/2}\right]\nonumber \\
 & \le4\sqrt{2}N\Big(\mathbb{E}\!\left[\mathbb{E}[\|\mathbf{h}\|_{2}^{2}\mid\mathcal{F}_{t-1}]\right]\Big)^{1/2}\Big(\mathbb{E}[\text{tr}(\bm{\Sigma}_{t-1})]\Big)^{1/2}\nonumber \\
 & =4\sqrt{2}NS\sqrt{\mathbb{E}[\text{tr}(\bm{\Sigma}_{t-1})]}.
\end{align}

Summing over $\ensuremath{t=1,2\ldots,T}$ proves $(\ref{eq:Bayesian_regret_decomposition}).$

\section{Proof of Corollary \ref{cor:Sublinear=000020Bayesian=000020regret}
~\protect\label{sec:proof_Sublinear=000020Bayesian=000020regret}}

Under $\mathbf{h}\sim\mathcal{CN}(\mathbf{0},\mathbf{I}_{N})$, the
prior covariance is $\bm{\Sigma}_{0}=\mathbf{I}_{N}$. The sequential
Gaussian observation model implies $\bm{\Sigma}_{t}=(\mathbf{I}_{N}+\sigma^{-2}\mathbf{G}_{t})^{-1}$.
Also, $S^{2}=\mathbb{E}[\|\mathbf{h}\|_{2}^{2}]=N$. We note that
condition (\ref{eq:condition_corollary_3}) is required only for $t\geq\bar{t}_{0}$.
That is because $\mathbf{G}_{t}$ is the sum of $t$ rank-one matrices
and hence cannot be full rank before the number of slots reaches the
channel dimension $N$. We split the regret bound in Theorem \ref{thm:Bayesian-regret-decomposition}
into the contribution of the first $\bar{t}_{0}$ slots and the contribution
of the remaining slots. Specifically, we first consider the contribution
of the first $\bar{t}_{0}$ slots. Since $\bm{\Sigma}_{t}\preceq\mathbf{I}_{N}$
for all $t\geq0$, we have $\text{tr}(\bm{\Sigma}_{t})\leq N$ for
all $t\geq0$. Therefore, 
\begin{equation}
\sum_{t=1}^{\min\{\bar{t}_{0},T\}}\sqrt{\mathbb{E}[\text{tr}(\bm{\Sigma}_{t-1})]}\leq\bar{t}_{0}\sqrt{N}
\end{equation}

According to (\ref{eq:Bayesian_regret_decomposition}), the total
regret accumulated over the first $\bar{t}_{0}$ is upper bounded
by a constant:

\begin{equation}
4\sqrt{2}NS\sum_{t=1}^{\min\{\bar{t}_{0},T\}}\sqrt{\mathbb{E}[\text{tr}(\bm{\Sigma}_{t-1})]}\le4\sqrt{2}NS\bar{t}_{0}\sqrt{N}.\label{eq:finite_initial_phase_bound}
\end{equation}

Then, consider the slot $t\geq\bar{t}_{0}+1$. For such $t$, condition
(\ref{eq:condition_corollary_3}) gives $\lambda_{\min}(\mathbf{G}_{t-1})\ge C((t-1)-\bar{t}_{0}+1)=C(t-\bar{t}_{0})$
almost surely. Hence, $\mathbf{G}_{t-1}\succeq C(t-\bar{t}_{0})\mathbf{I}_{N}.$
Substituting this lower bound into the $\bm{\Sigma}_{t-1}$ gives,
for all $t\geq\bar{t}_{0}+1$,

\begin{equation}
\bm{\Sigma}_{t-1}=(\mathbf{I}_{N}+\sigma^{-2}\mathbf{G}_{t-1})^{-1}\preceq\frac{1}{1+\sigma^{-2}C(t-\bar{t}_{0})}\mathbf{I}_{N}.
\end{equation}

Therefore,
\begin{equation}
\text{tr}(\bm{\Sigma}_{t-1})\le\frac{N}{1+\sigma^{-2}C(t-\bar{t}_{0})},\,\forall t\geq\bar{t}_{0}+1\label{eq:trace_sigma_tminus1_bound_post_t0}
\end{equation}

Since this bound holds almost surely, it also holds after taking expectation:
\begin{equation}
\mathbb{E}[\text{tr}(\bm{\Sigma}_{t-1})]\le\frac{N}{1+\sigma^{-2}C(t-\bar{t}_{0})},\,\forall t\geq\bar{t}_{0}+1.\label{eq:expected_trace_sigma_tminus1_bound_post_t0}
\end{equation}

Substituting (\ref{eq:finite_initial_phase_bound}) and (\ref{eq:expected_trace_sigma_tminus1_bound_post_t0})
into the Bayesian regret decomposition in Theorem \ref{thm:Bayesian-regret-decomposition}
gives
\begin{equation}
\mathcal{R}_{T}\le4\sqrt{2}NS\bar{t}_{0}\sqrt{N}+4\sqrt{2}NS\sum_{t=\bar{t}_{0}+1}^{T}\frac{\sqrt{N}}{\sqrt{1+\sigma^{-2}C(t-\bar{t}_{0})}}.\label{eq:RT_two_parts_cor1}
\end{equation}

Let $a\triangleq\sigma^{-2}C$. Since $x\mapsto(1+ax)^{-1/2}$ is
positive and decreasing on $x>0$, we have
\begin{align}
\sum_{t=\bar{t}_{0}+1}^{T} & \frac{1}{\sqrt{1+\sigma^{-2}C(t-\bar{t}_{0})}}\nonumber \\
 & \le1+\int_{0}^{T-\bar{t}_{0}}\frac{dx}{\sqrt{1+\sigma^{-2}Cx}}\nonumber \\
 & =1+\frac{2\sigma^{2}}{C}\bigl(\sqrt{1+\sigma^{-2}C(T-\bar{t}_{0})}-1\bigr)\nonumber \\
 & =O(\sqrt{T}).\label{eq:integral_test_cor1}
\end{align}

Combining (\ref{eq:RT_two_parts_cor1}) and (\ref{eq:integral_test_cor1})
shows that $\mathcal{R}_{T}=O(\sqrt{T})$ with $T\rightarrow\infty$.
When $t_{0}\leq N$, we have $\bar{t_{0}}=N$. Hence, the regret bound
becomes 
\[
4\sqrt{2}NS\bar{t}_{0}\sqrt{N}=4\sqrt{2}N^{5/2}
\]
which equals $4\sqrt{2}N^{3}$ under the Rayleigh prior $S=\sqrt{N}$.

\section{Proof of Theorem \ref{thm:sublinear=000020comm=000020with=000020channel=000020excitation}~\protect\label{sec:proof=000020sublinear=000020comm=000020with=000020channel=000020excitation} }

We define $\mathbf{S}_{t}\triangleq\sum_{s=1}^{t}\mathbf{Q}_{s}$,
$\mathbf{G}_{t}\triangleq\sum_{s=1}^{t}\mathbf{b}_{s}^{*}\mathbf{b}_{s}^{\top}$
and $\mathbf{X}_{t}\triangleq\mathbf{b}_{t}^{*}\mathbf{b}_{t}^{\top}-\mathbf{Q}_{t}$.
Under the complex-vector convention used in this paper, $\mathbf{b}_{t}^{*}\mathbf{b}_{t}^{\top}$
is Hermitian positive semidefinite because, for any $\mathbf{x}\in\mathbb{C}^{N}$,
$\mathbf{x}^{\mathrm{H}}\mathbf{b}_{t}^{*}\mathbf{b}_{t}^{\top}\mathbf{x}=|\mathbf{b}_{t}^{\top}\mathbf{x}|^{2}\ge0$.
Since $\mathbf{Q}_{t}=\mathbb{E}[\mathbf{b}_{t}^{*}\mathbf{b}_{t}^{\top}\mid\mathcal{F}_{t-1}]$,
the sequence $\{\mathbf{X}_{t}\}_{t\ge1}$ is a Hermitian matrix martingale
difference sequence.

Since each entry of $\mathbf{b}_{t}$ has unit modulus, we have $\|\mathbf{b}_{t}^{*}\mathbf{b}_{t}^{\top}\|=N$.
Moreover, since $\mathbf{b}_{t}^{*}\mathbf{b}_{t}^{\top}\preceq N\mathbf{I}_{N}$,
its conditional expectation also satisfies $\mathbf{Q}_{t}\preceq N\mathbf{I}_{N}$
and hence $||\mathbf{Q}_{t}||\leq N$. Therefore, 
\begin{equation}
\|\mathbf{X}_{t}\|\le\|\mathbf{b}_{t}^{*}\mathbf{b}_{t}^{\top}\|+\|\mathbf{Q}_{t}\|\le2N.
\end{equation}

It follows that
\begin{equation}
\mathbf{X}_{t}^{2}\preceq\|\mathbf{X}_{t}\|^{2}\mathbf{I}_{N}\preceq4N^{2}\mathbf{I}_{N}.
\end{equation}

By the matrix Azuma inequality for self-adjoint matrix martingales
as stated in \cite{tropp2012user}, for any $u>0$,
\begin{equation}
\mathbb{P}\left(\lambda_{\min}\left(\sum_{s=1}^{t}\mathbf{X}_{s}\right)\le-u\right)\le N\exp\left(-\frac{u^{2}}{32N^{2}t}\right).
\end{equation}
Taking $u=(c_{\alpha}/2)t^{\alpha}$ yields
\begin{align}
\mathbb{P} & \left(\lambda_{\min}\left(\sum_{s=1}^{t}\mathbf{X}_{s}\right)\le-\frac{c_{\alpha}}{2}t^{\alpha}\right)\nonumber \\
 & \hspace*{1em}\le N\exp\left(-\frac{c_{\alpha}^{2}}{128N^{2}}t^{2\alpha-1}\right).\label{eq:matrix_azuma_poly}
\end{align}

Let $\mathcal{A}_{t}\triangleq\left\{ \lambda_{\min}(\mathbf{S}_{t})\ge c_{\alpha}t^{\alpha}\right\} $
and $\mathcal{B}_{t}\triangleq\left\{ \lambda_{\min}(\mathbf{G}_{t}-\mathbf{S}_{t})\ge-(c_{\alpha}/2)t^{\alpha}\right\} .$
By assumption, $\mathbb{P}(\mathcal{A}_{t}^{c})\le q_{t,\alpha}$.
The concentration bound in $(\ref{eq:matrix_azuma_poly})$ gives
\begin{equation}
\mathbb{P}(\mathcal{B}_{t}^{c})\le N\exp\left(-\frac{c_{\alpha}^{2}}{128N^{2}}t^{2\alpha-1}\right).
\end{equation}

On the event $\mathcal{A}_{t}\cap\mathcal{B}_{t}$, Weyl's inequality
gives
\begin{equation}
\lambda_{\min}(\mathbf{G}_{t})\ge\lambda_{\min}(\mathbf{S}_{t})+\lambda_{\min}(\mathbf{G}_{t}-\mathbf{S}_{t})\ge\frac{c_{\alpha}}{2}t^{\alpha}.
\end{equation}

On this event, the posterior covariance satisfies
\begin{equation}
\text{tr}(\bm{\Sigma}_{t})=\text{tr}\left((\mathbf{I}_{N}+\sigma^{-2}\mathbf{G}_{t})^{-1}\right)\le\frac{N}{1+\frac{c_{\alpha}}{2\sigma^{2}}t^{\alpha}}.
\end{equation}
On the complementary event, we use the trivial bound $\text{tr}(\bm{\Sigma}_{t})\le N$.
Therefore,
\begin{align}
\mathbb{E}[\text{tr}(\bm{\Sigma}_{t})]\le & \frac{N}{1+\frac{c_{\alpha}}{2\sigma^{2}}t^{\alpha}}+N\mathbb{P}(\mathcal{A}_{t}^{c})+N\mathbb{P}(\mathcal{B}_{t}^{c})\nonumber \\
\le & \frac{N}{1+\frac{c_{\alpha}}{2\sigma^{2}}t^{\alpha}}+Nq_{t,\alpha}\nonumber \\
 & \,+N^{2}\exp\left(-\frac{c_{\alpha}^{2}}{128N^{2}}t^{2\alpha-1}\right).\label{eq:app_trace_bound_poly}
\end{align}

Substituting (\ref{eq:app_trace_bound_poly}) into the Bayesian regret
decomposition in Theorem \ref{thm:Bayesian-regret-decomposition}.
Since $\mathbf{h}\sim\mathcal{CN}(\mathbf{0},\mathbf{I}_{N})$, we
have $S^{2}=\mathbb{E}[||\mathbf{h}||_{2}^{2}]=N$. Moreover, $\bm{\Sigma}_{0}=\mathbf{I}_{N}$,
and hence 
\begin{equation}
\sum_{t=1}^{T}\sqrt{\mathbb{E}[\text{tr}(\bm{\Sigma}_{t-1})]}=\sqrt{N}+\sum_{t=1}^{T-1}\sqrt{\mathbb{E}[\text{tr}(\bm{\Sigma}_{t-1})]}.\label{eq:seperation_trace}
\end{equation}

Substituting (\ref{eq:seperation_trace}) into (\ref{eq:app_trace_bound_poly})
gives 
\begin{align}
\mathcal{R}_{T}=4\sqrt{2} & N^{3/2}\biggl(\sqrt{N}+\sum_{t=1}^{T-1}\Bigl(\frac{N}{1+\frac{c_{\alpha}}{2\sigma^{2}}t^{\alpha}}+Nq_{t,\alpha}\nonumber \\
 & ++N^{2}\exp\bigl(-\frac{c_{\alpha}^{2}}{128N^{2}}t^{2\alpha-1}\bigr)\Bigr){}^{1/2}\biggr)
\end{align}
which is the stated finite-horizon bound.

Using $\sqrt{a+b+c}\le\sqrt{a}+\sqrt{b}+\sqrt{c}$, we obtain
\begin{align}
\mathcal{R}_{T} & =O\left(\sum_{t=1}^{T}t^{-\alpha/2}\right)+O\left(\sum_{t=1}^{T}\sqrt{q_{t,\alpha}}\right)\nonumber \\
 & \quad+O\left(\sum_{t=1}^{T}\exp\left(-\frac{c_{\alpha}^{2}}{256N^{2}}t^{2\alpha-1}\right)\right).
\end{align}

Since $\alpha>1/2$, the exponential series is summable. Also, $\sum_{t=1}^{T}t^{-\alpha/2}=O(T^{1-\alpha/2})=o(T)$
for any $\alpha>0$. Hence, if $\sum_{t=1}^{T}\sqrt{q_{t,\alpha}}=o(T)$,
then $\mathcal{R}_{T}/T\to0$ as $T\to\infty$. 

\section{Proof of Proposition \ref{prop:unified-dict-vector}~\protect\label{sec:unified-dict-vector-proof}}

For the $n$th RIS element, the steering vector in (\ref{eq:steering=000020vector})
satisfies 
\begin{equation}
[\mathbf{s}(\vartheta,r)]_{n}=\frac{1}{\sqrt{N}}\exp\left\{ j\pi\left[-n\vartheta+\frac{n^{2}d}{2r}(1-\vartheta^{2})\right]\right\} .
\end{equation}

Therefore,
\begin{align}
 & [\mathbf{s}(\vartheta_{\mathrm{b}},r_{\mathrm{b}})\odot\mathbf{s}(\vartheta_{\mathrm{u}},r_{\mathrm{u}})]_{n}\nonumber \\
 & =\frac{1}{N}\exp\{j\pi[-n(\vartheta_{\mathrm{b}}+\vartheta_{\mathrm{u}})+\frac{n^{2}d}{2}(\frac{1-\vartheta_{\mathrm{b}}^{2}}{r_{\mathrm{b}}}+\frac{1-\vartheta_{\mathrm{u}}^{2}}{r_{\mathrm{u}}})]\}.\label{eq:proof_unified_product_raw}
\end{align}

Let
\begin{equation}
\nu=\vartheta_{\mathrm{b}}+\vartheta_{\mathrm{u}},\qquad\kappa=\frac{1-\vartheta_{\mathrm{b}}^{2}}{r_{\mathrm{b}}}+\frac{1-\vartheta_{\mathrm{u}}^{2}}{r_{\mathrm{u}}}.\label{eq:kappa_equal}
\end{equation}

The folding rule gives $\vartheta_{{\rm eq}}-\nu\in\{-2,0,2\}$ in
(\ref{eq:folding}). Therefore, for each $n=0,1,\ldots,N-1$, 
\begin{equation}
e^{-j\pi n\vartheta_{{\rm eq}}}=e^{-j\pi n\nu}e^{-j\pi n(\vartheta_{{\rm eq}}-\nu)}=e^{-j\pi n\nu}.\label{eq:var_eq}
\end{equation}

Substituting equivalent representative (\ref{eq:kappa_equal}) and
(\ref{eq:var_eq}) into (\ref{eq:proof_unified_product_raw}) gives
\begin{equation}
[\mathbf{s}(\vartheta_{\mathrm{b}},r_{\mathrm{b}})\odot\mathbf{s}(\vartheta_{\mathrm{u}},r_{\mathrm{u}})]_{n}=\frac{1}{N}\exp\left\{ j\pi\left[-n\vartheta_{{\rm eq}}+\frac{n^{2}d}{2}\kappa\right]\right\} .\label{eq:including_kappa}
\end{equation}
If $\kappa=0$, then (\ref{eq:including_kappa}) equals $1/\sqrt{N}[\mathbf{s}(\vartheta_{\text{eq}},\infty)]_{n}.$
If $\kappa>0$ and $|\vartheta_{\text{eq}}|<1$, then the definition
of $r_{{\rm eq}}$, gives $\kappa=(1-\vartheta_{{\rm eq}}^{2})/r_{{\rm eq}}$.
Therefore,
\begin{align}
[\mathbf{s} & (\vartheta_{\mathrm{b}},r_{\mathrm{b}})\odot\mathbf{s}(\vartheta_{\mathrm{u}},r_{\mathrm{u}})]_{n}\nonumber \\
 & =\frac{1}{N}\exp\left\{ j\pi\left[-n\vartheta_{{\rm eq}}+\frac{n^{2}d}{2r_{{\rm eq}}}(1-\vartheta_{{\rm eq}}^{2})\right]\right\} \nonumber \\
 & =\frac{1}{\sqrt{N}}[\mathbf{s}(\vartheta_{{\rm eq}},r_{{\rm eq}})]_{n}.
\end{align}

Since this holds for every $n=0,1,\ldots,N-1$, the vector identity
in (\ref{eq:unified_half_wavelength_product}) follows. Applying this
identity to every BS-RIS/RIS-user path pair in (\ref{eq:multipath-equiv-chann})
gives (\ref{eq:unified_sparse_cascaded_channel}).

\section{Proof of Proposition \ref{prop:ring-angule}~\protect\label{sec:proof-ring-ang}}

Let $\omega(\vartheta)\triangleq1-\vartheta^{2}$ and define the normalized
inverse-range variable 
\begin{equation}
\mu\triangleq\frac{\omega(\vartheta)}{r}.\label{eq:definite_mu_r}
\end{equation}

We first analyze the range-boundary rule. Let $\kappa_{\delta}\triangleq4\beta_{\delta}^{2}$.
From (\ref{eq:range-equation-2}), the offset from a dictionary center
to its auxiliary $\delta$-power boundary in the $\mu$-domain is
$\kappa_{\delta}/Z$, which is determined by $\beta_{\delta}$ and
is independent of $\vartheta$. Since two neighboring dictionary atoms
share the same auxiliary $\delta$-power boundary, the center-to-center
inverse-range step is $2\kappa_{\delta}/Z$. Substituting $r=\omega(\vartheta)/\mu$
into the update rule $r^{+}=rZ\omega(\vartheta)/(Z\omega(\vartheta)+2\kappa_{\delta}r),$
we obtain the transformed update in the $\mu$-domain: 
\begin{equation}
\mu^{+}=\mu+\frac{2\kappa_{\delta}}{Z}.\label{eq:mu_plus}
\end{equation}

Thus, the range recursion corresponds to a linear arithmetic progression
in $\mu$, with a constant center-to-center step size $2\kappa_{\delta}/Z$
independent of the angle $\vartheta$.

Consider the range sequence $\{r_{a,q}\}$ at angle $\vartheta_{a}$.
Let $q=0$ denote the far-field center, i.e., $\mu_{a,0}=0$. Based
on (\ref{eq:mu_plus}), the center $q$ is explicitly given by 
\begin{equation}
\mu_{a,q}=q\frac{2\kappa_{\delta}}{Z},\qquad q=0,1,\ldots.\label{eq:mu_a_q}
\end{equation}

Next, consider the range-first sequence given by $r_{m,q}^{\text{RF}}=\omega(\vartheta_{m})r_{a,q}/\omega(\vartheta_{a})$,
with $r_{m,0}^{\text{RF}}=\infty$. For $q\ge1$, applying the definition
(\ref{eq:definite_mu_r}), its inverse variable transforms as 
\begin{equation}
\mu_{m,q}^{\text{RF}}=\frac{\omega(\vartheta_{m})}{r_{m,q}^{\text{RF}}}=\frac{\omega(\vartheta_{m})}{\frac{\omega(\vartheta_{m})}{\omega(\vartheta_{a})}r_{a,q}}=\frac{\omega(\vartheta_{a})}{r_{a,q}}=\mu_{a,q}\label{eq:mu_ring_a_q}
\end{equation}

For $q=0$, the same equality holds by the far-field convention $r=\infty$
and $\mu=0$. Substituting (\ref{eq:mu_a_q}) into (\ref{eq:mu_ring_a_q})
yields $\mu_{m,q}^{\text{RF}}=2q\kappa_{\delta}/Z$.

Now consider the angular-first sequence $\{r_{m,q}^{\text{AF}}\}$
generated at $\vartheta_{m}$, the far-field center has $\mu_{m,0}^{\text{AF}}=0$.
Since the update in the $\mu$-domain is angle independent, the angular-first
atoms satisfy $\mu_{m,q}^{\text{AF}}=2q\kappa_{\delta}/Z$.

Therefore, for every retained index pair $(m,q)$, $\mu_{m,q}^{\text{AF}}=\mu_{m,q}^{\text{RF}}$
and hence $r_{m,q}^{\text{RF}}=r_{m,q}^{\text{AF}}$, and for $q=0,$
the same conclusion holds by the far-field convention $r=\infty$.

It remains to verify that the shared $\delta$-power boundaries are
also preserved. For each retained neighboring pair, the common auxiliary
boundary lies one offset $\kappa_{\delta}/Z$ away from each center
in the $\mu$-domain. Let $\widetilde{r}_{m,q}^{{\rm RF}}$ denote
the boundary between $\ensuremath{r_{m,q}^{{\rm RF}}}$ and $r_{m,q+1}^{{\rm RF}}$,
defined by $\omega(\vartheta_{m})/\widetilde{r}_{m,q}^{\mathrm{RF}}=\mu_{m,q}^{\mathrm{RF}}+\kappa_{\delta}/Z=(2q+1)\kappa_{\delta}/Z$.
Hence, under the Fresnel range-boundary approximation,

\begin{align}
G(\vartheta_{m},\widetilde{r}_{m,q}^{{\rm RF}};\vartheta_{m},r_{m,q}^{{\rm RF}}) & =\delta,\\
G(\vartheta_{m},\widetilde{r}_{m,q}^{{\rm RF}};\vartheta_{m},r_{m,q+1}^{{\rm RF}}) & =\delta.\nonumber 
\end{align}

Therefore, the range-first construction has the same centers and the
same auxiliary $\delta$-power boundaries as the angular-first construction
for all retained index pairs. This proves the sampling-order equivalence.
\end{document}